\def\isarxiv{1} 

\ifdefined\isarxiv
\documentclass[11pt]{article}

\usepackage[numbers]{natbib}

\else
\documentclass{article}
\usepackage{neurips_2022}
\fi

\usepackage{amsmath}
\usepackage{amsthm}
\usepackage{amssymb}
\usepackage{algorithm}
\usepackage{subfig}
\usepackage{algpseudocode}
\usepackage{graphicx}
\usepackage{grffile}
\usepackage{wrapfig,epsfig}
\usepackage{url}
\usepackage{xcolor}
\usepackage{epstopdf}

\usepackage{bbm}
\usepackage{dsfont}

\allowdisplaybreaks

\ifdefined\isarxiv

\usepackage{tikz}
\usepackage{hyperref}  
\hypersetup{colorlinks=true,citecolor=blue,linkcolor=blue} 
\usetikzlibrary{arrows}
\usepackage[margin=1in]{geometry}

\else

\usepackage{microtype}
\usepackage{hyperref}
\definecolor{mydarkblue}{rgb}{0,0.08,0.45}
\hypersetup{colorlinks=true, citecolor=mydarkblue,linkcolor=mydarkblue}

\fi

\theoremstyle{plain}
\newtheorem{theorem}{Theorem}[section]
\newtheorem{lemma}[theorem]{Lemma}
\newtheorem{definition}[theorem]{Definition}

\newtheorem{proposition}[theorem]{Proposition}

\newtheorem{fact}[theorem]{Fact}

\newcommand{\wh}{\widehat}
\newcommand{\wt}{\widetilde}

\newcommand{\R}{\mathbb{R}}

\renewcommand{\hat}{\wh}

\DeclareMathOperator{\poly}{poly}
\DeclareMathOperator{\polylog}{polylog}

\DeclareMathOperator{\nnz}{nnz}

\DeclareMathOperator{\rank}{rank}
\DeclareMathOperator{\diag}{diag}

\DeclareMathOperator{\vol}{vol}

\makeatletter
\newcommand*{\RN}[1]{\expandafter\@slowromancap\romannumeral #1@}
\makeatother

\usepackage{lineno}

\begin{document}

\ifdefined\isarxiv

\date{}

\title{Quantum Speedups for Approximating the John Ellipsoid}
\author{
Xiaoyu Li\thanks{\texttt{
xli216@stevens.edu}. Stevens Institute of Technology.}
\and 
Zhao Song\thanks{\texttt{ zsong@adobe.com, magic.linuxkde@gmail.com}. Adobe Research.}
\and 
Junwei Yu\thanks{\texttt{ yujunwei04@berkeley.edu}. University of California, Berkeley.}
}

\else

\title{Intern Project} 
\maketitle 
\fi

\ifdefined\isarxiv
\begin{titlepage}
  \maketitle
  \begin{abstract}
In 1948, Fritz John proposed a theorem stating that every convex body has a unique maximal volume inscribed ellipsoid, known as the John ellipsoid. The John ellipsoid has become fundamental in mathematics, with extensive applications in high-dimensional sampling, linear programming, and machine learning. Designing faster algorithms to compute the John ellipsoid is therefore an important and emerging problem. In [Cohen, Cousins, Lee, Yang COLT 2019], they established an algorithm for approximating the John ellipsoid for a symmetric convex polytope defined by a matrix $A \in \mathbb{R}^{n \times d}$ with a time complexity of $O(nd^2)$. This was later improved to $O(\text{nnz}(A) + d^\omega)$ by [Song, Yang, Yang, Zhou 2022], where $\nnz(A)$ is the number of nonzero entries of $A$ and $\omega$ is the matrix multiplication exponent. Currently $\omega \approx 2.371$ [Alman, Duan, Williams, Xu, Xu, Zhou 2024]. In this work, we present the first quantum algorithm that computes the John ellipsoid utilizing recent advances in quantum algorithms for spectral approximation and leverage score approximation, running in $O(\sqrt{n}d^{1.5} + d^\omega)$ time. In the tall matrix regime, our algorithm achieves quadratic speedup, resulting in a sublinear running time and significantly outperforming the current best classical algorithms. 

  \end{abstract}
  \thispagestyle{empty}
\end{titlepage}

{\hypersetup{linkcolor=black}
\tableofcontents
}
\newpage

\else

\begin{abstract}
In 1948, Fritz John proposed a theorem stating that every convex body has a unique maximal volume inscribed ellipsoid, known as the John ellipsoid. The John ellipsoid has become fundamental in mathematics, with extensive applications in high-dimensional sampling, linear programming, and machine learning. Designing faster algorithms to compute the John ellipsoid is therefore an important and emerging problem. In [Cohen, Cousins, Lee, Yang COLT 2019], they established an algorithm for approximating the John ellipsoid for a symmetric convex polytope defined by a matrix $A \in \mathbb{R}^{n \times d}$ with a time complexity of $O(nd^2)$. This was later improved to $O(\text{nnz}(A) + d^\omega)$ by [Song, Yang, Yang, Zhou 2022], where $\nnz(A)$ is the number of nonzero entries of $A$ and $\omega$ is the matrix multiplication exponent. Currently $\omega \approx 2.371$ [Alman, Duan, Williams, Xu, Xu, Zhou 2024]. In this work, we present the first quantum algorithm that computes the John ellipsoid utilizing recent advances in quantum algorithms for spectral approximation and leverage score approximation, running in $O(\sqrt{n}d^{1.5} + d^\omega)$ time. In the tall matrix regime, our algorithm achieves quadratic speedup, resulting in a sublinear running time and significantly outperforming the current best classical algorithms. 
\end{abstract}

\fi

\section{Introduction}\label{sec:intro}

The concept of the John ellipsoid, introduced in Fritz John's seminal work~\cite{john1948extremum}, plays a fundamental role in convex geometry and functional analysis~\cite{ball1991volume, lutwak1993brunn, barthe1998extremal, ball2001convex}. It states that any convex body (a compact convex set $K$ in $\R^d$ with a nonempty interior) in $\R^d$ contains a unique inscribed ellipsoid $E$ with the maximum volume, known as the John ellipsoid of $K$. Mathematically, the John ellipsoid of $K$ is the unique solution to the following optimization problem:
\begin{align*}
    \max_{E \in \mathcal E^d} &~ \vol(E)\ \mathrm{s.t.}\ E \subseteq K
\end{align*}
where $\vol(E)$ is the volume of the ellipsoid $E$ and $\mathcal E^d$ is the set of all ellipsoid in $\R^d$. Furthermore, $K$ is contained within a dilation of $E$ by a factor of $d$, i.e. $K \subseteq d \cdot E$. If $K$ is symmetric (i.e., $K = -K$), this can be improved to $K \subseteq \sqrt{d} \cdot E$. 

Due to numerous useful properties of the John ellipsoid, it has extensive applications in various fields such as high-dimensional sampling~\cite{vempala2005geometric, chen2018fast}, linear programming~\cite{ls14}, bandit algorithms~\cite{bubeck2012towards, hazan2016volumetric}, differential privacy~\cite{nikolov2013geometry}, inverse problems~\cite{lin2020nonnegative} and control theory~\cite{halder2020smallest, tang2024uncertainty}. Furthermore, \cite{todd2016minimum} demonstrated that determining the John Ellipsoid of a symmetric polytope is equivalent to solving the well-known D-optimal experiment design problem~\cite{atwood1969optimal}. This problem has numerous applications in machine learning~\cite{allen2017near, wang2017computationally, lu2018relatively}.

In this paper, we study the case of finding John ellipsoid $E$ of a symmetric convex polytope $P := \{x \in \R^d: -\mathbf{1}_n \leq Ax \leq \mathbf{1}_n\}$ where $A \in \R^{n\times d}$ is a tall matrix, i.e., $n \gg d$, and $\mathbf{1}_n$ is an all-ones vector in $\R^d$. In this setting, \cite{cohen2015lp} showed that finding the John ellipsoid is equivalent to computing $\ell_\infty$ Lewis weights of $A$. The $\ell_\infty$ Lewis weight of $A$ is the unique vector $w = (w_1, \ldots, w_n) \in \R^n$ satisfying $w_i = w_ia_i^\top (A^\top W A)a_i$ for each $i \in [n]$, where $a_i^\top$ is the $i$-th row of $A$ and $W = \diag(w)$. Since $\ell_\infty$ Lewis weight naturally satisfies a fixed point equation, \cite{cohen2019near} developed an algorithm using a fixed point iteration in time $\wt{O}(\epsilon^{-1} nd^2)$ where $\epsilon$ is the approximation error. This was latter improved by~\cite{song2022faster} using leverage score sampling and achieved running time $\wt{O}(\epsilon^{-2}\nnz(A) + \epsilon^{-3}d^\omega)$ where $\nnz(A)$ is the number of nonzero entries of $A$ and $\omega$ is the exponent matrix multiplication~\cite{williams2012multiplying, le2014powers, alman2021refined, williams2024new, alman2024asymmetryyieldsfastermatrix}. Currently, $\omega \approx 2.372$~\cite{alman2024asymmetryyieldsfastermatrix}. Since $\nnz(A) = \Omega(n)$, the current state-of-the-art requires $\Omega(n)$ time, and this is the best possible for classical algorithms since computing the weighted leverage scores of $A$ takes at least $\Omega(\nnz(A) + d^\omega)$ time~\cite{nelson2013osnap, clarkson2017low}.

Recently, it has been demonstrated that quantum algorithms can achieve provable speedups for solving many well-known foundational problems such as linear regression~\cite{wang2017quantum, gilyen2022improved, song2023revisitingquantumalgorithmslinear}, linear programming~\cite{casares2020quantum, apers2023quantum}, semi-definite programming~\cite{brandao2017quantum, brandao2019quantum, mohammadisiahroudi2023quantum}, graph algorithms~\cite{apers2022quantum, ben2020symmetries}, high-dimensional sampling~\cite{wocjan2008speedup, montanaro2015quantum, wild2021quantum} and attention computation~\cite{gao2023fast}. In this work, we aim to leverage the advantages of quantum algorithms to accelerate the computation of the John ellipsoid. This leads us to the following question:
\begin{center}
    \textit{Is it possible to find the John ellipsoid of a symmetric convex polytope in $O_{d, \epsilon}(\sqrt{n})$\footnote{We use $O_g(f(n))$ to denote $f(n) \cdot g^{O(1)}$.} time?}
\end{center}

In this work, we provide a positive answer to this question. 
Our approach draws inspiration from recent research conducted by \cite{apers2023quantum}, which introduced quantum algorithms for spectral approximation. These algorithms were effectively utilized to enhance the performance of interior point methods, providing a significant acceleration in solving linear programming problems. We extend their methods to compute the $\ell_\infty$ Lewis weights by exploiting its fixed-point property. Additionally, in the regime where the matrix $A$ is tall, our algorithm achieves a quadratic speedup, significantly outperforming the current best classical state-of-the-art.

\subsection{Our Results}\label{sub:our_res}

Our main contribution is the first quantum algorithm that can compute the John ellipsoid of a symmetric convex polytope that achieves a quadratic speedup and leads to a sublinear running time in $n$. We state our main result informally as follows.

\begin{theorem}[Main result, informal version of Theorem~\ref{thm:main}]\label{thm:informal}
    Given $A \in \R^{n \times d}$, let $P := \{x \in \R^d : -\mathbf{1}_n \leq Ax \leq \mathbf{1}_n\}$ be a symmetric convex polytope.
    For all $\epsilon \in (0, 0.5)$, there is a quantum algorithm that outputs an ellipsoid $E$ satisfies 
    \begin{itemize}
        \item $ \frac{1}{\sqrt{1+\epsilon}} \cdot E \subseteq P \subseteq \sqrt{(1+\epsilon)d} \cdot E. $
        \item $\vol(\frac{1}{\sqrt{1+\epsilon}} \cdot E) \geq e^{-\epsilon d} \cdot \vol(E^*)$ where $E^*$ is the exact John ellipsoid of $P$.
        \item It runs in $\Tilde{O}(\epsilon^{-2}\sqrt{n}d^{1.5}+\epsilon^{-3}d^\omega)$ time.
        \item It makes $\Tilde{O}(\epsilon^{-2}\sqrt{nd})$ queries to $A$.
        \item The success probability is $1 - 1/\poly(n)$.
    \end{itemize}
\end{theorem}

Instead of using traditional convex optimization algorithms, our algorithm (Algorithm~\ref{alg:main}) follows the fixed point iteration framework developed in~\cite{cohen2019near}, which provides a simple and clean convergence analysis with very fast convergence rate. We developed the quantum fixed point iteration for computing $\ell_\infty$ Lewis weight and use it as a subroutine in our algorithm. 

Given our interest in the tall matrix regime, i.e., $n \gg d$, our primary objective is to reduce the algorithm's running time dependence on $n$. We remark that our algorithm can be easily extended to the sparse setting and become more efficient. When the matrix $A$ is sparse, we can further improve the running time dependence on $d$.

We also explored a tensor generalization of John ellipsoid computation, where the polytope is defined by the Kronecker product. In this case, we develop an algorithm (Algorithm~\ref{alg:main_tensor}) can use Algorithm~\ref{alg:main} as a subroutine to solve this problem. We state our results informally as follows.
\begin{theorem}[Main result, informal version of Theorem~\ref{thm:main_tensor}]\label{thm:main_tensor_informal}
Given $A_1, A_2 \in \R^{n \times d}$, let $P := \{x \in \R^{d^2} : -\mathbf{1}_{n^2} \leq (A_1 \otimes A_2) x \leq \mathbf{1}_{n^2}\}$ be a symmetric convex polytope.
    For all $\epsilon \in (0, 0.5)$, there is a quantum algorithm that outputs an ellipsoid $E$ satisfies  
    \begin{itemize}
        \item $ \frac{1}{1+\epsilon} \cdot E \subseteq P \subseteq (1+\epsilon)d \cdot E. $
        \item $\vol(\frac{1}{1+\epsilon}E) \geq e^{-\epsilon^2 d^2} \cdot \vol(E^*)$ where $E^*$ is the exact John ellipsoid of $P$.
        \item It runs in $\wt{O}(\epsilon^{-2}\sqrt{n}d^{1.5}+\epsilon^{-3}d^\omega)$ time.
        \item It makes $\Tilde{O}(\epsilon^{-2}\sqrt{nd})$ queries to $A_1$ and $A_2$.
        \item The success probability is $1 - 1/\poly(n)$.
    \end{itemize}
\end{theorem}

\paragraph{Roadmap.} 
The rest of the paper is organized as follows. In Section~\ref{sec:rela}, we present related work. In Section~\ref{sec:prel}, we outline the preliminary notations and mathematical tools used in our analysis. In Section~\ref{sec:prob}, we provide the problem formulation, including the background knowledge.
In Section~\ref{sec:main_alg}, we describe our main algorithm and performance guarantees. In Section~\ref{sec:tensor},  we generalize the John ellipsoid computation to a tensor version, where the polytope is defined by the Kronecker product. In Section~\ref{sec:conc}, we conclude with a summary of our contributions and suggestions for future research directions.
\section{Related Work}\label{sec:rela}

\subsection{Faster John Ellipsoid Computation}
The properties and structures of the John ellipsoid have been extensively studied in the mathematical community~\cite{ball1991volume, lutwak1993brunn, barthe1998extremal, ball2001convex, lutwak2005john, zou2014orlicz, todd2016minimum}. However, the computationally efficient methods for computing the John Ellipsoid remain largely unexplored.
Since the computation of the John ellipsoid can be formulated as a convex optimization problem, the most straightforward approach is to apply convex optimization algorithms. Historically, before the advances made in \cite{cohen2019near}, the main algorithms for computing the John ellipsoid were predominantly based on first-order methods \cite{khachiyan1996rounding, kumar2005minimum, damla2008linear} and second-order interior point methods \cite{nesterov1994interior, sun2004computation}. These methods leveraged the mathematical properties of convex optimization to iteratively approach the solution. \cite{cohen2019near} proposed an innovative fixed-point iteration framework that utilized the structure of the optimization problem. In addition, they developed a faster algorithm that incorporated sketching techniques, which provided an efficient way to handle large-scale data by reducing dimensionality while preserving essential geometric properties. Building upon this work, \cite{song2022faster} introduced further improvements by integrating leverage score sampling with sketching techniques and also provided a analysis for scenarios where the matrix $A$ exhibited a low tree-width structure. Recently, \cite{lls+24} studied fast John ellipsoid computation under the differential privacy setting.

\subsection{Leverage Scores and Lewis Weights} 
The computation of the John ellipsoid is closely related to leverage scores and Lewis weights, which quantify the statistical significance of a matrix's rows and facilitate the selection of critical data points. Both leverage scores and Lewis weights have various applications in many different areas. One major application is in graph theory and combinatorial optimization. For example, \cite{blss20, liu2020breaking} used leverage scores to calculate the maximum match. \cite{madry2013navigating, madry2016computing, liu2020faster} explored the application of leverage score in the max-flow problems. Other applications include graph sparsification~\cite{spielman2008graph, gsy23_hyper}, and generating random spanning trees~\cite{schild2018almost}. In optimization, they are utilized to design faster algorithms for solving linear programming~\cite{ls14, blss20} and semi-definite programming~\cite{jkl+20}. In matrix analysis, they are explored in tensor decomposition~\cite{song2019relative} and matrix decomposition~\cite{boutsidis2014optimal, song2019relative}. Additional works address the approximation of leverage scores~\cite{spielman2008graph, drineas2012fast}, while the generalized concept of Lewis weights is examined by~\cite{bourgain1989approximation, cohen2015lp}.

\subsection{Linear Programming and Semi-Definite Programming}
Linear programming constitutes a fundamental area of study within computer science and optimization theory. The Simplex algorithm, introduced in \cite{dan51}, is a pivotal technique in this domain, despite its exponential runtime. \cite{k79} introduced the Ellipsoid method, which offers a polynomial-time solution to linear programming. The ellipsoid method is theoretically significant but often outperformed by the Simplex method in practical applications.  A noteworthy advancement in the field is the interior-point method \cite{k84}, which combines polynomial-time efficiency with robust practical performance, particularly in real-world scenarios. This methodological innovation has catalyzed extensive research efforts aimed at accelerating the interior-point method for solving a broad spectrum of classical optimization problems. Interior point method has a wide impact on linear programming as well as other complex tasks, studied in \cite{v87, r88, v89, ds08, ls14, cls19, lsz19, b20, ls20a, blss20, jswz21, sy21, gs22}. The contributions of John Ellipsoid have been especially significant in interior point method. For example,  John Ellipsoid is utilized to find path to solutions in \cite{ls14}. Furthermore, the interior point method, along with the theoretical underpinnings provided by John Ellipsoid, plays a crucial role in addressing semidefinite programming challenges, such as \cite{jkl+20, syz23, gs22, huang2022faster, huang2022solving}. Linear programming and semidefinite programming are extensively utilized in the field of machine learning theory, particularly in areas such as empirical risk minimization \cite{lsz19, sxz23, qszz23} and support vector machines \cite{gsz23, gswy23}.

\subsection{Quantum Sketching Algorithms}  
Sketching, a powerful technique for dimensionality reduction, has been extensively studied over the last decades~\cite{w14, nelson2013osnap, clarkson2017low, dssw18, swyz21, sy21, swyz23}. In quantum computing, the block-encoding framework often renders sketching unnecessary due to its ability to exponentially reduce dimensions. For example, it has been show that a quantum computer can solve linear regression very quickly without using any sketching technique~\cite{wang2017quantum, cgj18, cd21, sha23}. However, there is some intrinsic limitations of block-encoding-based methods, so the exploration of quantum sketching algorithms becomes a popular topic. \cite{apers2022quantum} proposed efficient quantum algorithms for solving graph sparsification and Laplacian systems with polynomial speedups. Additionally, quantum algorithms have been used for speeding up leverage score sampling, and quantum algorithms that  are dependent on data-specific parameters like the matrix's condition number, have been developed~\cite{prakash2014quantum, liu2017fast, sha23}. Further advancements include a quantum algorithm for solving linear programming with a quantum matrix spectral approximation~\cite{apers2023quantum} and its extension to Kronecker spectral approximation~\cite{gao2024quantumspeedupspectralapproximation}. 
\section{Preliminary}\label{sec:prel}

\subsection{Notations}\label{sub:nota}
Throughout the paper, we assume $A \in \R^{n \times d}$ with $n \gg d$ and that $A$ is of full rank. We use $\nnz(A)$ to denote the number of nonzero entries in $A$. For a vector $v \in \R^n$, we denote its $\ell_2$ norm by $\| v \|_2$. For a vector $v \in \R^n$, we use $\diag(u)$ to denote the diagonal matrix such that $\diag(u)_{ii} = u_i$. For two matrices $A, B$ of the same size, the Hadamard product of $A$ and $B$ is denoted by $A \circ B$ where $(A \circ B)_{i,j} = A_{i,j} B_{i,j}$. We say that a vector $w \in \R^n_{\geq 0}$ if $w$ is an $n$-dimensional vector with nonnegative entries. We use $\mathbf{1}_n$ to denote an all-one vector in $\R^n$. For a function $f$, we use $\widetilde{O}(f)$ to hide polylogarithmic factors, represented as $O(f) \cdot \polylog(f)$. For an integer $n$, $[n]$ denotes the set $\{1, 2, \ldots, n\}$. For two scalars $a$ and $b$, we use $a = (1\pm \epsilon) b$ to denote $(1-\epsilon)b \leq a \leq (1+\epsilon)b$. For two vectors $u, v \in \R^n$, we use $u = (1\pm \epsilon) v$ to denote $(1-\epsilon)v_i \leq u_i \leq (1+\epsilon)v_i$ for every $i \in [n]$. We use $\omega$ to denote the matrix multiplication exponent, i.e., multiplying two $n\times n$ matrices takes time $O(n^\omega)$. Currently $\omega \approx 2.371$~\cite{alman2024asymmetryyieldsfastermatrix}.

\subsection{Basic Tools}
We state a useful tool from linear algebra. This decomposition is particularly useful in various optimization and machine learning algorithms, where it enables efficient computation and analysis of weighted sums of outer products.

\begin{fact}[Rank-1 Decomposition]\label{fact:rank1}
Given matrix $A \in \R^{n \times d}$ and a vector $w \in \R^n$, we can decompose
\begin{align*}
    A^\top W A = \sum_{i=1}^n w_i a_i a_i^\top
\end{align*}
where $W = \diag(w)$ and $a_i^\top$ is the $i$-th row of $A$.
\end{fact}

\subsection{Leverage Score}\label{sub:lev_score}
The leverage scores can be defined using several equivalently ways as follows:
\begin{definition}[Leverage score]\label{def:lev_score}
    Given a matrix $A \in \R^{n\times d}$, let $U \in \R^{n\times d}$ be an orthonormal basis for the column space of $A$. For any $i \in [n]$, the leverage score of $i$-th row of $A$ can be equivalently defined as:
    \begin{itemize}
        \item $\sigma_i(A) = \| u_i \|_2$.
        \item $\sigma_i(A) = a_i^\top (A^\top A)^{-1} a_i$. 
        \item $\sigma_i(A) = \max_{x \in \R^d} (a_i^\top x)^2 / \|Ax\|_2^2$.
    \end{itemize}
\end{definition}

The last definition provides an intuitive explanation of leverage scores. The leverage score for row $a_i$ is higher if it is more significant, indicating that we can identify a vector $\mathbf{x}$ that has a large inner product with $a_i$ compared to its average inner product (i.e., $\|A\mathbf{x}\|_2^2$) with all other rows in the matrix. This insight underlies a common technique known as leverage score sampling, where rows with higher leverage scores are sampled with higher probability. It is used by~\cite{song2022faster} to speed up the computation of John ellipsoid.

Next, we summarized some well-known properties of leverage scores and state the classical leverage score approximation lemma from~\cite{deng2022discrepancy}.

\begin{proposition}[Forklore]\label{prop:lev_fork}
Let $A \in \R^{n \times d}$ be a matrix.
\begin{itemize}
    \item For every $i \in [n]$, we have $\sigma_i(A) \in [0,1]$.
    \item $\sum_{i=1}^n \sigma_i(A) = d$.
\end{itemize}
\end{proposition}

\begin{lemma}[Classical leverage score approximation, Lemma~4.3 of \cite{deng2022discrepancy}]\label{lem:cla_lev_score}
Given a matrix $A \in \R^{n \times d}$, for any $\epsilon, \delta \in (0,1)$, there is a classical algorithm that outputs $\Tilde{\sigma}$ satisfying
\begin{itemize}
    \item $\Tilde{\sigma}$ is a $(1+\epsilon)$-approximation of $\sigma_i(A)$, i.e., $\Tilde{\sigma} = (1 \pm \epsilon)\sigma_i(A)$.
    \item It takes $\Tilde{O}(\epsilon^{-2}(\nnz(A)+d^\omega))$ time where $\Tilde{O}$ hides $\log(n/\delta)$ factor.
    \item The success probability is $1-\delta$.
\end{itemize}

\end{lemma}

We can observe that leverage scores can be approximated significantly faster when the matrix \( A \) is sufficiently sparse. However, when $ A $ is a tall matrix, since $\text{nnz}(A) = \Omega(n)$, it still requires $\widetilde{O}(\epsilon^{-2}n)$ time to compute a $(1+\epsilon)$-approximation of the leverage scores. To over the linear dependence on $n$, we need some tools from quantum computing.

\subsection{Tools from Quantum Computing}\label{sub:tool_quan}

In this paper, we assume the standard quantum computational model, where the system can execute quantum subroutines on $O(\log n)$ qubits, performs quantum queries to the input, and accesses a quantum-read/classical-write RAM with at most $\poly(n)$ bits. For our problem, the quantum query is a row query to the given matrix $A$.

We begin by outlining the well-known Grover's search algorithm~\cite{grover1996fast}, which serves as a fundamental building block for numerous modern quantum algorithms. Grover's algorithm is designed for unstructured search problems and efficiently locates, with high probability, the unique input to a black box function that yields a specified output value. Remarkably, it achieves quadratic speedups.

\begin{lemma}[Grover's search, \cite{grover1996fast}]
If the following conditions hold:
\begin{itemize}
    \item There is an oracle $\mathcal O$ evaluating the function $g : [n] \to \{0,1\}$.
    \item Let $f^{-1}(1) = \{ i \in [n] | f(i) = 1\}$
    \item $|f^{-1}(1)| = k$ for some unknown integer $k \leq n$.
\end{itemize}
Then there is a quantum algorithm such that
\begin{itemize}
    \item It finds all $i$'s in $f^{-1}(i)$ by making $\wt{O}(\sqrt{nk})$ queries to $\mathcal O$.
    \item Moreover, if each $f(i)$ requires $\mathcal T$ to compute, we find all $i$'s in $f^{-1}(i)$ in $\wt{O}(\sqrt{nk}\mathcal T)$ time.
    \item The success probability is $1/2$.
\end{itemize}
    
\end{lemma}

Next, we state quantum tools for fast leverage score computation from \cite{apers2023quantum}. The quantum spectral approximation algorithm proposed by~\cite{apers2023quantum} is based on a simple recursive sparsification algorithm in~\cite{cohen2015uniform}. Given a matrix $A \in \mathbb{R}^{n \times d}$ where $n \gg d$, there is a quantum algorithm that efficiently approximates $A^\top A$. The state-of-the-art classical approach has $O(\nnz(A) + d^\omega)$ running time~\cite{nelson2013osnap, clarkson2017low}. Essentially, this quantum speedup of $n \to \sqrt{nd}$ is achieved by utilizing Grover's search algorithm to locate the $O(d)$ significant rows among the $n$ rows of $B$ with $O(\sqrt{nd})$ row queries to $B$.

\begin{lemma}[Quantum spectral approximation, Theorem~3.1 of \cite{apers2023quantum}]\label{lem:qua_spe_score}
Suppose that we have the query access to a matrix $A \in \R^{n \times d}$. Then for any $\epsilon \in (0,1)$, there exists a quantum algorithm that outputs a matrix $B \in \R^{\Tilde{O}(d/\epsilon^2)\times d}$ satisfying

\begin{itemize}
    \item $(1-\epsilon)B^\top B \preceq A^\top A \preceq(1+\epsilon)B^\top B$.
    \item It makes $\Tilde{O}(\epsilon^{-1}\sqrt{nd})$ queries to $A$
    \item It takes $\Tilde{O}(\epsilon^{-1}\sqrt{n}d^{1.5} + d^\omega)$ time. 
    \item The success probability $1 - 1/\poly(n)$. 
\end{itemize}
\end{lemma}

Based on Lemma~\ref{lem:qua_spe_score}, \cite{apers2023quantum} also proposed an algorithm that efficiently approximate the leverage scores. This is the main tool we will use later to achieve quantum speedup for approximating the John ellipsoid because computing the $\ell_\infty$ Lewis weight is roughly just the weighted version of leverage score. Similar to Lemma~\ref{lem:qua_spe_score}, the quantum leverage score approximation algorithm achieved a speedup of $n \to \sqrt{nd}$ for tall dense matrices.
\begin{lemma}[Quantum leverage score approximation, Theorem~3.2 of \cite{apers2023quantum}]\label{lem:qua_lev_score}
Suppose that we have the query access to a matrix $A \in \R^{n \times d}$. Then for any $\epsilon \in (0,1)$, there exists a quantum algorithm that outputs an approximation $\Tilde{\sigma}$ of leverage score $\sigma(A)$ satisfying
\begin{itemize}
    \item $\Tilde{\sigma}$ is a $(1+\epsilon)$-approximation of $\sigma(A)$, i.e., $\Tilde{\sigma} = (1 \pm \epsilon)\sigma(A)$.
    \item It makes $\Tilde{O}(\epsilon^{-1}\sqrt{nd})$ queries to $A$.
    \item It takes $\Tilde{O}(\epsilon^{-1}\sqrt{n}d^{1.5} + \epsilon^{-2} d^\omega)$ time.
    \item The success probability $1 - 1/\poly(n)$.
\end{itemize}
\end{lemma}
\section{Problem Formulation}\label{sec:prob}

In this section, we formally formulate the problem of the John Ellipsoid computation as a convex optimization problem. We first define several main concepts.

\subsection{Background on Polytope and Ellipsoid}

\begin{definition}[Symmetric convex polytope]\label{def:sym_covx_polt}
We define the symmetric convex polytope as a convex set $P \subseteq \R^d$ that takes the form of
\begin{align*}
    P := &~ \{ x \in \R^d : |\langle a_i, x \rangle| \leq 1, \forall i \in [n]\}.
\end{align*}
where $a_i^\top$ is the $i$-th row of a matrix $A \in \R^{n\times d}$.
\end{definition}

We assume that $A$ is full rank. The symmetry of $P$ implies that any maximum volume ellipsoid inside $P$ must be centered at origin. 

\begin{definition}[Origin-centered ellipsoid]\label{def:orig_cen_ell}
An origin-centered ellipsoid is a set $E \subseteq \R^d$ takes the form of
\begin{align*}
    E := \{ x \in \R^d : x^\top G^{-2} x \leq 1\}
\end{align*}
where $G \in \R^{d \times d}$ is a positive semi-definite matrix.
\end{definition}

\subsection{Primal and Dual Formulation} \label{sub:primal_dual}

By simple linear algebra, we know that $\vol(E)$ is proportional to $\det G$. In \cite{cohen2019near}, it is shown that an ellipsoid $E$ is a subset of a polytope $P$ if and only if for every $i \in [n]$, $\|G a_i\|_2 \leq 1$. Therefore, we formulate the John ellipsoid computation problem as the following optimization problem.

\begin{definition} [Primal program, implicitly in Section~2 of \cite{cohen2019near}]
\label{def:primal_classic}
Let $A$ be defined in Definition~\ref{def:sym_covx_polt}.
Let $G$ be defined in Definition~\ref{def:orig_cen_ell}. The optimization problem of John Ellipsoid computation could be formulated as
\begin{align}\label{eq:john_ell_opt}
    \max_{G} &~ \log \det G^2, \\
    \text{s.t.} &~ \| G a_i \|_2 \leq 1, \forall i \in [n], \notag \\
    &~ G \succeq 0. \notag
\end{align}
\end{definition}

\begin{lemma}[Dual program, implicitly in Section~2 of \cite{cohen2019near}]
\label{lem:dual_classics}
The dual program of program~\eqref{eq:john_ell_opt} in Definition~\ref{def:primal_classic} is
\begin{align}\label{eq:dual_opt}
    \min_{w} &~ \sum_{i=1}^n w_i - \log \det(\sum_{j=1}^n w_j a_j a_j^\top) - d, \\
    \text{s.t.} &~ w_i \geq 0, \forall i \in [n]. \notag
\end{align}
\end{lemma}

\begin{proof}
Firstly, we derive the Lagrangian
\begin{align}
\mathcal{L}(G, w) = \log \det G^2 + \sum_{i=1}^n w_i (1 - \|G a_i\|_2^2). \label{eq:lag_classic}
\end{align}
By setting the gradient with respect to $G$ to zero, we have
\begin{align*}
\nabla_G \mathcal{L} = 2G^{-1} -2\sum_{i=1}^n w_i G a_i a_i^\top = 0.
\end{align*}
Hence, we have
\begin{align}
G^{-2} = \sum_{i=1}^n w_i a_i a_i^\top. \label{eq:G_-2_optimal}
\end{align}

Let $B \in \R^{n \times d}$ be a matrix such that $i$-th row of $B$ is $b_i^\top = \sqrt{w_i}a_i^\top$ for $i \in [n]$. Then we have
\begin{align}
    \sum_{i=1}^n w_i \|Ga_i\|_2^2 
    = & ~ \sum_{i=1}^n w_i a_i^\top G^2 a_i \notag \\
    = & ~ \sum_{i=1}^n w_i a_i^\top(\sum_{j=1}^n w_j a_j a_j^\top)^{-1} a_i \notag \\
    = & ~ \sum_{i=1}^n (\sqrt{w_i} a_i)^\top(\sum_{j=1}^n (\sqrt{w_j} a_j)(\sqrt{w_j} a_j)^\top)^{-1} (\sqrt{w_i} a_i) \notag \\
    = & ~ \sum_{i=1}^n b_i^\top(B^\top B)^{-1} b_i \notag \\
    = & ~ d \label{eq:wi_G_ai}
\end{align}
where the first step comes from the definition of norm, the second step comes from Eq.~\eqref{eq:G_-2_optimal}, the third step follows by basic algebra, the fourth step uses the definition of $B$, and the last step comes from Proposition~\ref{prop:lev_fork}.

By plugging Eq.~\eqref{eq:G_-2_optimal} into Eq.~\eqref{eq:lag_classic}, we have
\begin{align*}
\mathcal{L}(G, w) = & ~ \log \det G^2 + \sum_{i=1}^n w_i (1 - \|G a_i\|_2^2) \\
= & ~ \log \det G^2 + \sum_{i=1}^n w_i - \sum_{i=1}^n w_i \|Ga_i\|_2^2 \\
= & ~ - \log \det G^{-2} + \sum_{i=1}^n w_i - \sum_{i=1}^n w_i \|Ga_i\|_2^2 \\
= & ~ - \log \det(\sum_{j=1}^n w_j a_j a_j^\top) + \sum_{i=1}^n w_i - \sum_{i=1}^n w_i \|Ga_i\|_2^2 \\
= & ~ \sum_{i=1}^n w_i - \log \det(\sum_{j=1}^n w_j a_j a_j^\top) - d
\end{align*}
where the first step comes from Eq.~\eqref{eq:lag_classic}, the second step follows from basic algebra, the third step comes from logarithm arithmetic, the fourth step comes from Eq.~\eqref{eq:G_-2_optimal}, and the last step comes from Eq.~\eqref{eq:wi_G_ai}.
\end{proof}

\subsection{Optimality Condition and Approximate Solution}
\label{sec:optimality}

The optimality condition of program~\eqref{eq:dual_opt} is given by \cite{todd2016minimum}:
\begin{lemma}[Optimality condition of program~\eqref{eq:dual_opt}, Proposition~2.5 of \cite{todd2016minimum}]\label{lem:opt_cond} 
A vector $w \in \R^n_{\geq 0}$ is an optimal solution to program~\eqref{eq:dual_opt} if and only if the following holds:
\begin{align*}
    \sum_{i=1}^n w_i = d, \\
    a_i^\top (\sum_{j=1}^n w_j a_j a_j^\top)^{-1} a_i = 1, \mathrm{if}~ w_i \neq 0, \\
    a_i^\top (\sum_{j=1}^n w_j a_j a_j^\top)^{-1} a_i < 1, \mathrm{if}~ w_i = 0.
\end{align*}
\end{lemma}

Instead of finding an exact solution, we are interested in approximate solutions to program~\eqref{eq:dual_opt} which is defined as follows:

\begin{definition}[$(1+\epsilon)$-approximate solution, variant of Definition~2.2 in \cite{cohen2019near}]\label{def:approx_sol}
For any $\epsilon > 0$, we say that $w \in \R^n_{\geq 0}$ is a $(1+\epsilon)$-approximate solution to program~\eqref{eq:dual_opt} if $w$ satisfies
\begin{align*}
    \sum_{i=1}^n w_i = (1 \pm \epsilon)d, \\
    a_i^\top (\sum_{j=1}^n w_j a_j a_j^\top)^{-1} a_i \leq 1 + \epsilon, \forall i \in [n].
\end{align*}
\end{definition}

Next, the following lemma states that an approximate solution is a good approximation of the exact John Ellipsoid. Recall that the exact John Ellipsoid $E^*$ satisfies $E^* \subseteq P \subseteq \sqrt{d} E^*$.
\begin{lemma}[Approximate John Ellipsoid, variant of Lemma 2.3 in~\cite{cohen2019near}]\label{lem:approx_john_ell} 
Let $P := \{x \in \R^d : - \mathbf 1_n \leq Ax \leq \mathbf 1_n\}$ be a symmetric convex polytope. Let $w \in \R^n_{\geq 0}$ be a $(1+\epsilon)$-approximate solution to program~\eqref{eq:dual_opt}. Let
\begin{align*}
    E := \{ x \in \R^d : x^\top (A^\top \diag(w) A) x \leq 1\}.
\end{align*}
Then 
\begin{align*}
    \frac{1}{\sqrt{1+\epsilon}} \cdot E \subseteq P \subseteq \sqrt{(1+\epsilon)d} \cdot E.
\end{align*}
Moreover,
\begin{align*}
    \vol(\frac{1}{\sqrt{1+\epsilon}}E) \geq e^{-\epsilon d} \cdot \vol(E^*).
\end{align*}
\end{lemma}

\begin{proof}
First, we show that $\frac{1}{\sqrt{1+\epsilon}}\cdot E\subseteq P$. Let $x \in \frac{1}{\sqrt{1+\epsilon}}\cdot E$. Then $x^\top G^{-2} x \leq \frac{1}{1+\epsilon}$ where $G^{-2} = A^\top \diag(w) A$. For every $i \in [n]$,
\begin{align*}
    (a_i^\top x)^2 = &~ ((Ga_i)^\top G^{-1}x)^2 \\
    \leq &~ \|Ga_i\|_2^2 \|G^{-1}x\|_2^2 \\
    \leq &~ (a_i^TG^2a_i)(x^\top G^{-2} x) \\ 
    \leq &~ (1+\epsilon) \cdot \frac{1}{1+\epsilon} = 1
\end{align*}
where the first step comes from $GG^{-1} = I$ and $G$ is symmetric, the second step is because of Cauchy-Schwarz inequality, the third step follows from writing norms as inner products, and the last step uses Definition~\ref{def:approx_sol} and $x \in \frac{1}{\sqrt{1+\epsilon}}\cdot E$.

Next, we show that $P \subseteq \sqrt{(1+\epsilon)d}\cdot E$.
Let $x \in P$. Then $|a_i^\top x| \leq 1$ for all $i \in [n]$. Hence
\begin{align*}
    x^\top (A^\top \diag(w) A) x = &~ x^\top (\sum_{i=1}^n w_i a_i^\top a_i) x \\
    = &~ \sum_{i=1}^n w_i (a_i^\top x)^2 \\
    \leq &~ \sum_{i=1}^n w_i \\
    \leq &~ (1+\epsilon)d
\end{align*}
where the first step uses Fact~\ref{fact:rank1}, the second step is a simple calculation, the third step follows from that for all $i \in [n]$ we have $|a_i^\top x| \leq 1$, and the last step uses Definition~\ref{def:approx_sol}.

Finally, we show that $\vol(\frac{1}{\sqrt{1+\epsilon}}E) \geq e^{-d\epsilon} \cdot \vol(E^*)$. Since $\frac{1}{\sqrt{1+\epsilon}}\cdot E\subseteq P$, $G' := ((1+\epsilon)A^\top \diag(w) A)^{-1/2}$ is a feasible solution to program~\eqref{eq:john_ell_opt} and $w$ is a feasible solution to program~\eqref{eq:dual_opt}, we have the following duality gap:
\begin{align}\label{eq:vol_derive}
     &~ (\sum_{i=1}^n w_i - \log \det(\sum_{i=1}^n w_i a_i a_i^\top) - d) - \log \det ((1+\epsilon)\sum_{i=1}^n w_i a_i a_i^\top )^{-1} \\
    \leq &~((1+\epsilon)d - \log \det(\sum_{i=1}^n w_i a_i a_i^\top) - d) - \log \det ((1+\epsilon)\sum_{i=1}^n w_i a_i a_i^\top )^{-1} \notag \\
    = &~((1+\epsilon)d - \log \det(\sum_{i=1}^n w_i a_i a_i^\top) - d) - \log ((1+\epsilon)^{-d} \det (\sum_{i=1}^n w_i a_i a_i^\top )^{-1}) \notag \\
    = &~ \epsilon d - \log (1+\epsilon)^{-d} \notag \\
    = &~ \epsilon d + d\log (1+\epsilon) \notag \\
    = &~ 2\epsilon d \notag 
\end{align}
where the first step uses Definition~\ref{def:approx_sol}, the second step uses the property of determinant, the third and fourth steps follow from simple algebra, the last step follows from $\log(1+x) \leq x$ for all $x\geq 0$. Let $E^* = \{x\in \R^d : x^\top (G^*)^{-2} x \leq 1\}$. Then we have
\begin{align}\label{eq:vol_derive_2}
    \log \det (G')^{2} \geq &~ \log\det (G^*)^2 - 2\epsilon d \\
    \geq &~ \log (\det (G^*)^2 e^{-2\epsilon d}) \notag \\ 
    \geq &~ \log (e^{-\epsilon d}\det (G^*))^2 \notag
\end{align}
where the first step follows from Eq.~\eqref{eq:vol_derive}, the second and last steps follow from basic algebra. By Eq.~\eqref{eq:vol_derive_2}, we have 
\begin{align}\label{eq:vol_derive_3}
    \det G' \geq e^{-\epsilon d} \det G^*.
\end{align}
Since the volume of a ellipsoid is proportional to its representation matrix, we have $\vol(\frac{1}{\sqrt{1+\epsilon}} \cdot E)$ is proportional to $\log G'$ and $\vol(E^*)$ is proportional to $\log G^*$. Then by Eq.\eqref{eq:vol_derive_3}, we have
\begin{align*}
    \vol(\frac{1}{\sqrt{1+\epsilon}}\cdot E) \geq e^{-\epsilon d} \vol(E^*). \tag*{\qedhere}
\end{align*}   
\end{proof}

Note that both Definition~\ref{def:approx_sol} and Lemma~\ref{lem:approx_john_ell} are a bit different from \cite{cohen2019near, song2022faster}. This is because their algorithms normalize the approximate solution to have the sum $d$. However, normalization require summing all entries, which needs at least $\Omega(n)$ time. To avoid this expensive step, we have to directly output the approximation with normalizing it. This leads to a slightly loose upper bound with factor $\sqrt{1+\epsilon}$.
\section{Main Algorithm}\label{sec:main_alg}
In this section, we present our fast quantum algorithm for computing the John ellipsoid. We begin by introducing the key component that accelerates the computation process: the quantum fixed point iteration. Following this, we establish a telescoping lemma, which is crucial for analyzing the correctness of our algorithm. Last, we will put everything together and show our main result.

\begin{algorithm*}[!t]
\caption{Quantum Algorithm for approximating John Ellipsoid}\label{alg:main} 
    \begin{algorithmic}[1]
    \Procedure{\textsc{ApproxJE}}{$A \in \R^{n \times d}, \epsilon \in (0,0.5)$} 
    \State \Comment{Symmetric convex polytope defined as $\{ x \in \R^d: -\mathbf{1}_n \leq Ax \leq \mathbf{1}_n\}$}
    \State \Comment{Target approximation error $\epsilon \in (0,1)$}
    \State $T \gets \Theta(\epsilon^{-1} \log(n/d))$ \Comment{Number of iterations}
    \State Initialize $w^{(1)}_i \gets d/n, \forall i \in [n]$ \Comment{Initialization}
    \For{$k=1 \to T$} 
        \State \Comment{Ideally we want to compute $\hat{w}^{(k+1)} = w^{(k)} \circ f(w^{(k)})$.} 
        \State $w^{(k+1)} \gets (1 \pm \epsilon) w^{(k)} \circ f(w^{(k)})$. \Comment{By Lemma~\ref{lem:qua_fix_point_iter}.}   
    \EndFor
    \State $w \gets \frac{1}{T}\sum_{k=1}^T w^{(k)}$ \Comment{By Lemma~\ref{lem:approx-guar}.}

    \State \Return $w$
    \EndProcedure
    \end{algorithmic}
\end{algorithm*}

\subsection{Quantum Fixed Point Iteration}\label{sub:quan_fix_point}
Given a matrix $A \in \R^{n\times d}$, let $f: \R^n \to \R^n$ be a function defined as $f(w) = (f_1(w), f_2(w), \ldots, f_n(w))$ such that for every $i \in [n]$, we have
\begin{align*}
    f_i(w) := a_i^\top (\sum_{j=1}^n w_j a_j a_j^\top)^{-1} a_i = a_i^\top(A^\top \diag(w) A)^{-1} a_i.
\end{align*}

Then we derive the following fixed point optimality condition of program~\eqref{eq:dual_opt}.
\begin{lemma}[Fixed point optimality of program~\eqref{eq:dual_opt}, implicitly in Section~2 of \cite{cohen2019near}]\label{lem:fix_point_opt}
A vector $w \in \R^n_{\geq 0}$ is an optimal solution to program~\eqref{eq:dual_opt} if and only if the following holds:
\begin{align*}
    \sum_{i=1}^n w_i = d, \\
    w_i=w_i f_i(w), \forall i \in [n].
\end{align*}
\end{lemma}

\begin{proof}
Suppose that $w \in \R^n_{\geq 0}$ is an optimal solution to program~\eqref{eq:dual_opt}, by Lemma~\ref{lem:opt_cond}, we have
\begin{align*}
    \sum_{i=1}^n w_i = d, \\
    f_i(w) = 1, \text{if } w_i \neq 0, \\
    f_i(w) < 1, \text{if } w_i = 0.
\end{align*}
For any $i \in [n]$ with $w_i \neq 0$, we have $f_i(w) = 1$, and hence $w_i=w_i f_i(w)$. For any $i \in [n]$ with $w_i = 0$, it clearly holds that $w_i=w_i f_i(w)$.

Conversely, suppose that the following holds:
\begin{align*}
    \sum_{i=1}^n w_i = d, \\
    w_i=w_i f_i(w), \forall i \in [n].
\end{align*}
For any $i \in [n]$, we have $w_i(1-f_i(w)) = 0$ by the second equation above. By Definition~\ref{def:lev_score}, $f_i(w) = \sigma_i(\sqrt{\diag(w)}A)$, and by Proposition~\ref{prop:lev_fork}, $f_i(w)\in (0,1)$. Thus, if $w_i \neq 0$, then $w_i(1-f_i(w)) = 0$ implies $f_i(w)= 1$. If $w_i = 0$, then $w_i(1-f_i(w)) = 0$ implies $f_i(w) < 1$. Therefore, the optimality of $w$ follows from Lemma~\ref{lem:opt_cond}.
\end{proof}

Inspired by Lemma~\ref{lem:fix_point_opt}, it is natural to consider the fixed point iteration scheme $w^{(k+1)} = w^{(k)} \circ f(w^{(k)})$. We will use quantum tool to speedup this procedure. However, we cannot directly compute $\sqrt{\diag(w^{(k)})}A$ and apply Lemma~\ref{lem:qua_lev_score} to compute its leverage score. This is because computing $\sqrt{\diag(w^{(k)})}A$ takes $O(nd)$ time, which is unaffordable for our purpose. Hence, we can only implicitly keep track of its entries.

Next, we have the following lemma which states that there is a quantum algorithm which can efficiently approximate the fixed point iteration, which is the key ingredient of our algorithm for computing the John ellipsoid.

\begin{lemma}[Quantum fixed point iteration]\label{lem:qua_fix_point_iter}
Suppose that we have the query access to a matrix $A \in \R^{n \times d}$ and a vector $w^{(k)} \in \R^n$, for every $\epsilon \in (0,1)$, there is a quantum algorithm that outputs a vector $w^{(k+1)}$ which is a $(1+\epsilon)$-approximation to original $
\hat{w}^{(k+1)} := w^{(k)} \circ f(w^{(k)})$ such that
\begin{itemize}
    \item For every $i \in [n]$, we have $w^{(k+1)}_i = (1\pm \epsilon)w^{(k)}_i f(w^{(k)}_i)$.
    \item It makes $\Tilde{O}(\epsilon^{-1}\sqrt{nd})$ queries to $A$ and $w^{(k)}$.
    \item It takes $\Tilde{O}(\epsilon^{-1}\sqrt{n}d^{1.5} + \epsilon^{-2}d^\omega)$ time.
    \item The success probability is $1 - 1/\poly(n)$.
\end{itemize}
\end{lemma}

\begin{proof}
Since we are given query access to $A$ and $w^{(k)}$, this directly implies that we can simulate the query access to $\sqrt{\diag(w^{(k)})}A$ because the $i$-th row of $\sqrt{\diag(w^{(k)})}A$ is $\sqrt{w^{(k)}_i}a_i^\top$. By Lemma~\ref{lem:qua_lev_score}, for any $\epsilon \in (0,1)$, there exists a quantum algorithm providing query access to
\begin{align*}
    \Tilde{\sigma}_i = (1 \pm \epsilon) \sigma_i(\sqrt{\diag(w^{(k)})}A).
\end{align*}
By simple algebraic operation, $f_i(w) = \sigma_i(\sqrt{\diag(w^{(k)})}A)$. Let 
\begin{align*}
    w^{(k+1)}_i := w^{(k)}_i \Tilde{\sigma}_i = (1\pm \epsilon)w^{(k)}_i f(w^{(k)}_i)
\end{align*} 
which is the desired output. The number of queries, running time and success probability follows from Lemma~\ref{lem:qua_lev_score}.
\end{proof}

\subsection{Telescoping Lemma}\label{sub:tele_lem}

The proof of our main result relies on the following lemma shown in~\cite{cohen2019near}, which states that for every $i \in [n]$, $\log f_i$ is a convex function.

\begin{lemma}[Convexity, Lemma~3.4 in \cite{cohen2019near}]\label{lem:convex}
For $i \in [n]$, we define $\phi_i : \R^n \to \R$ as follows:
\begin{align*}
    \phi_i(w) := \log f_i(w) = \log(a_i^\top (\sum_{i=1}^m w_i a_i a_i^\top)^{-1} a_i).
\end{align*}
Then $\phi_i$ is convex.
\end{lemma}

Since $\phi_i$ is a convex function for every $i \in [n]$, \cite{cohen2019near} applies Jensen's inequality to show the following telescoping lemma, which will be used to analyze the correctness of our algorithm.

\begin{lemma}[Telescoping, Lemma~C.4 of~\cite{cohen2019near}]\label{lem:tele} 
Let $T$ be the number of iterations executed in Algorithm~\ref{alg:main}. Let $w$ be the output of Algorithm~1. For any $i \in [n]$, it holds that
\begin{align*}
    \phi_i(w) \leq \frac{1}{T}\log \frac{n}{d} + \frac{1}{T}\sum_{k=1}^T\log \frac{\hat{w}^{(k)}_i}{w^{(k)}_i} 
\end{align*}
\end{lemma}

\begin{lemma}[Telescoping by Quantum Fixed Point Iteration]\label{lem:tele_quant} 
Let $T$ be the number of iterations executed in Algorithm~\ref{alg:main}. Let $w$ be the output of Algorithm~1. For any $\epsilon \in (0, 1)$, with probability $1 - 1/\poly(n)$, for any $i \in [n]$, it holds that
\begin{align*}
    \phi_i(w) \leq \frac{1}{T} \log (\frac{n}{d}) + \log(1+\epsilon).
\end{align*}
\end{lemma}
\begin{proof}

By Lemma~\ref{lem:tele_quant}, for any $i \in [n]$, we have
\begin{align*}
    \phi_i(w) \leq \frac{1}{T}\log \frac{n}{d} + \frac{1}{T}\sum_{k=1}^T\log \frac{\hat{w}^{(k)}_i}{w^{(k)}_i} .
\end{align*}

We only need to bound the second term on the right-hand side. For sufficiently large $n$, by union bound and Lemma~\ref{lem:qua_fix_point_iter}, we have 
\begin{align}\label{eq:uni-bound}
    \Pr[w^{(k)} = (1 \pm \epsilon) \hat{w}^{(k)}, \forall k \in [T]] \geq 1/\poly(n).
\end{align}

Hence, with probability $1 - 1/\poly(n)$,
\begin{align*}
    \frac{1}{T}\sum_{k=1}^T\log \frac{\hat{w}^{(k)}_i}{w^{(k)}_i} \leq  &~ \frac{1}{T}\sum_{k=1}^T\log \frac{\hat{w}^{(k)}_i}{(1-\epsilon)\hat{w}^{(k)}_i} \\
    = &~ \frac{1}{T}\sum_{k=1}^T \log \frac{1}{1-\epsilon} \\
    = &~ \log \frac{1}{1-\epsilon} \\
    \leq &~ \log (1+2\epsilon).
\end{align*}
where the first step simply follows from Eq.~\eqref{eq:uni-bound} in the previous proofs, the second and third step come from basic algebra, and the last step follows from the fact that the inequality $1/(1-\epsilon) \leq 1+2\epsilon$ holds for any $\epsilon \in [0, 0.5]$.
\end{proof}

\subsection{Approximation Guarantee and Main Results}\label{sub:perf_guar}

Now we can put everything together and show our main result (Theorem~\ref{thm:main}). 
First, we demonstrate that Algorithm~\ref{alg:main} outputs a $(1+\epsilon)$-approximate solution to program~\eqref{eq:dual_opt} when the number of iterations $T$ is appropriately chosen, i.e., $T = \Theta(\epsilon^{-1}\log(n/d))$.

\begin{lemma}[Approximation guarantee]\label{lem:approx-guar}
For any $\epsilon \in (0,0.5)$, with probability $1 - 1/\poly(n)$, Algorithm~\ref{alg:main} produces an $(1+\epsilon)$-approximate solution.
\end{lemma}
\begin{proof}
Let $w$ be the output of Algorithm~\ref{alg:main}.
By Eq.~\eqref{eq:uni-bound}, we have 
\begin{align}
\label{eq:w_approx}
    w = \frac{1}{T}\sum_{k=1}^T w^{(k)} \leq \frac{1}{T}\sum_{k=1}^T (1+\epsilon)\hat{w}^{(k)}.
\end{align}
with probability $1 - 1/\poly(n)$.

Hence,
\begin{align}
\label{eq:sum_w_approx}
    \sum_{i=1}^n w_i = &~ \sum_{i=1}^n\frac{1}{T}\sum_{k=1}^T w^{(k)}_i \notag \\
    \leq &~ \sum_{i=1}^n \frac{1}{T} \sum_{k=1}^T (1+\epsilon)\hat{w}^{(k)}_i \notag \\
    = &~ (1+\epsilon) \sum_{k=1}^T \frac{1}{T}\sum_{i=1}^n \hat{w}^{(k)}_i \notag \\
    = &~ (1+\epsilon) \sum_{k=1}^T \frac{1}{T} \notag \\
    = &~ 1+\epsilon.
\end{align}
where the first step comes from $w = \frac{1}{T}\sum_{k=1}^T w^{(k)}$, the second step uses Eq.~\eqref{eq:w_approx}, the third step follows from rearranging, the fourth step is because of Lemma~\ref{prop:lev_fork}, and the last step simply follows from basic calculation.

Since $T = 2\epsilon^{-1}\log(n/d)$, By Lemma~\ref{lem:tele}, for $\epsilon \in (0,1)$, we have
\begin{align}
\label{eq:phi_i_w_approx}
    \phi_i(w) \leq &~ \frac{1}{T}\log(\frac{n}{d}) + \log(1+\epsilon) \notag \\
    \leq &~ \frac{\epsilon}{2} + \log(1+\epsilon) \notag \\
    \leq &~ \log(1+\epsilon) + \log(1+\epsilon) \notag \\
    = &~ \log(1+\epsilon)^2 \notag\\
    \leq  &~ \log (1+3\epsilon)
\end{align}
where the first step comes from Lemma~\ref{lem:tele}, the second step follows from $T = 2\epsilon^{-1}\log(n/d)$, the third step comes from $\epsilon/2 \leq \log(1+\epsilon)$ for all $\epsilon \in (0,1)$, the fourth step comes from basic algebra, the fifth step uses simple algebra, and the last step comes the fact that the inequality $(1+\epsilon)^2 \leq 1+3\epsilon$ holds for all $\epsilon \in (0, 0.25)$.

By Lemma~\ref{lem:convex} and
Eq.~\eqref{eq:phi_i_w_approx}, for any $\epsilon \in (0, 0.5)$, we have 
\begin{align}\label{eq:lev_approx}
    a_i^\top(\sum_{j=1}^n w_j a_j a_j^\top) a_i \leq 1+\epsilon.
\end{align}

Finally, Definition~\ref{def:approx_sol}, Eq.~\eqref{eq:w_approx} and Eq.~\eqref{eq:lev_approx} imply that $w$ is an $(1+\epsilon)$-approximate solution to program~\eqref{eq:dual_opt}.
\end{proof}

Putting everything together, we can show our main result.

\begin{theorem}[Main result]\label{thm:main}
    Given $A \in \R^{n \times d}$, let $P := \{x \in \R^d : -\mathbf{1}_n \leq Ax \leq \mathbf{1}_n\}$ be a symmetric convex polytope.
    For all $\epsilon \in (0,0.5)$, Algorithm~\ref{alg:main} outputs $w \in \R^n$ such that the ellipsoid $E := \{x \in \R^d: x^\top(A^\top\diag(w)A)x \leq 1\}$ satisfies 
    \begin{itemize}
        \item $ \frac{1}{\sqrt{1+\epsilon}} \cdot E \subseteq P \subseteq \sqrt{(1+\epsilon)d} \cdot E. $
        \item $\vol(\frac{1}{\sqrt{1+\epsilon}}E) \geq e^{-\epsilon d} \cdot \vol(E^*)$ where $E^*$ is the exact John ellipsoid of $P$.
        \item It runs in $\wt{O}((\epsilon^{-2}\sqrt{n}d^{1.5}+\epsilon^{-3}d^\omega)$ time.
        \item It makes $\Tilde{O}(\epsilon^{-2}\sqrt{nd})$ queries to $A$.
        \item The success probability is $1 - 1/\poly(n)$.
    \end{itemize}
\end{theorem}

\begin{proof}
    The correctness of main result directly follows from Lemma~\ref{lem:approx-guar} and Lemma~\ref{lem:approx_john_ell}.
    The time complexity and number of queries follow from Lemma~\ref{lem:qua_fix_point_iter}.
\end{proof}

\section{Tensor Generalization of John Ellipsoid}\label{sec:tensor}
In this section, we consider a tensor generalization of the John ellipsoid, where the polytope is constructed by the Kronecker product. This problem can be solved in a similar fashion we have developed before, but instead we need the quantum speedup for spectral approximation of the Kronecker product~\cite{gao2024quantumspeedupspectralapproximation}. However, in this section, we show that this tensor version of John ellipsoid computation problem can be decomposed into two ordinary John ellipsoid computation problems and call Algorithm~\ref{alg:main} as a subroutine to solve it. 

\subsection{Kronecker Product}
We state the definition of Kronecker product and some basic properties of it.
\begin{definition}
    For two matrices $A \in R^{n_1 \times d_1}, B \in R^{n_2 \times d_2}$, we use $A \otimes B$ to denote the Kronecker product of $A$ and $B$ where 
    \begin{align*}
        (A \otimes B)_{i_1 + (i_2-1)n_1, j_1+(j_2-1)d_1} = A_{i_1,j_1}B_{i_2,j_2}.
    \end{align*}
\end{definition}
\begin{fact} [Properties of Kronecker product] \label{fact:kronecker}
Let $A, B, C, D \in \R^{n \times d}$. Let $w_1, w_2 \in \R^n$. Then, we have the following properties
\begin{itemize}
    \item $A \otimes (B + C) = A \otimes B + A \otimes C$
    \item $(B + C) \otimes A = B \otimes A + C \otimes A$
    \item $(A \otimes B) \otimes C = A \otimes (B \otimes C)$
    \item $(A \otimes B) (C \otimes D) = (AC) \otimes (BD)$
    \item $(A \otimes B)^\top = A^\top \otimes B^\top$
    \item $(A \otimes B)^{-1} = A^{-1} \otimes B^{-1}$
    \item $\diag (w_1 \otimes w_2) = \diag (w_1) \otimes \diag (w_2)$ 
\end{itemize}
\end{fact}

\subsection{Problem Formulation}
\begin{definition}[Symmetric convex polytope defined by Kronecker product] \label{def:sym_covx_polt_tensor_final}
We define the symmetric convex polytope as a convex set $P \subseteq \R^{d^2}$ that takes the form of
\begin{align*}
    P := &~ \{ x \in \R^{d^2} : |\langle a_{1,i_1} \otimes a_{2,i_2}, x \rangle| \leq 1, \forall i,j \in [n]\}.
\end{align*}
where $a_{1,i_1}^\top$ and $a_{2,i_2}^\top$ are $i_1$-th row of $A_1 \in \R^{n\times d}$ and $i_2$-th row of $A_2 \in \R^{n \times d}$, respectively.
\end{definition}

We assume $A_1$ and $A_2$ are full rank.
We also define $\mathsf{A} := A_1 \otimes A_2$. Then the $i_1 + n(i_2-1)$-th row in $\mathsf{A}$, denoted by $\mathsf{a}_{i_1+(i_2-1)n}^\top$, is equal to $ a_{1,i_1}^\top \otimes a_{2,i_2}^\top$. Note that $\rank(\mathsf{A}) = \rank(A_1) \cdot \rank(A_2)$, so $\mathsf{A}$ is also full rank. Similarly as before, we define the origin-centered ellipsoid as follows.

\begin{definition}[Origin-centered ellipsoid]\label{def:orig_cen_ell_tensor_final}
An origin-centered ellipsoid is a set $E \subseteq \R^{d^2}$ takes the form of
\begin{align*}
    E := \{ x \in \R^{d^2} : x^\top G^{-2} x \leq 1\}
\end{align*}
where $G \in \R^{d^2}$ is a positive semidefinite matrix.
\end{definition}

\begin{definition} [Primal program, tensor version of Definition~\ref{def:primal_classic}]
\label{def:primal_tensor_final}
Let $A_1,A_2$ be defined in Definition~\ref{def:sym_covx_polt_tensor_final} and let $\mathsf{A} = A_1 \otimes A_2$. Let $G$ be defined in Definition~\ref{def:orig_cen_ell_tensor_final}. The optimization problem of John Ellipsoid computation could be formulated as
\begin{align} \label{eq:je_primal_tensor_final_2}
    \max_{G} &~ \log \det G^2, \\
    \mathrm{s.t.} &~ \| G (a_{1,i_1} \otimes a_{2, i_2})\|_2 \leq 1, \forall i_1, i_2 \in [n], \notag \\
    &~ G \succeq 0 \notag,
\end{align}
or equivalently,
\begin{align} \label{eq:je_primal_tensor_final}
    \max_{G} &~ \log \det G^2, \\
    \mathrm{s.t.} &~ \| G \mathsf{a}_{i} \|_2 \leq 1, \forall i \in [n^2], \notag \\
    &~ G \succeq 0 \notag.
\end{align}
\end{definition}

\begin{lemma}[Dual program, tensor version of Lemma~\ref{lem:dual_classics}]
\label{lem:dual_tensor_final}
The dual program of program~\eqref{eq:je_primal_tensor_final} is
\begin{align}
\label{eq:je_dual_tensor}
\min_\mathsf{w} & ~ \sum_{i=1}^{n^2} \mathsf{w}_i - \log \det (\sum_{i=1}^{n^2} \mathsf{w}_{i} \mathsf{a}_{i} \mathsf{a}_{i}^\top) - d^2, \\
\mathrm{s.t.} & ~ \mathsf{w}_i \geq 0, \forall i \in [n^2] \notag.
\end{align}
\end{lemma}

We remark that our results can be generalized to $\mathsf A$ is defined by the Kronecker product of $m$ matrices, i.e., $\mathsf{A} = A_1 \otimes A_2 \otimes \cdots \otimes A_m$. As long as $m \leq O(1)$, our algorithm will eventually achieve the same running time. For simplicity, we only present the case of $\mathsf{A} = A_1 \otimes A_2$.

\subsection{Optimality Conditions and Approximate Solution}
\label{sec:optimality_tensor}

\begin{lemma}[Optimality condition of program~\eqref{eq:je_dual_tensor}, tensor version of Lemma~\ref{lem:opt_cond}] \label{lem:opt_cond_tensor} 
A vector $w \in \R^{n^2}_{\geq 0}$ is an optimal solution to program~\eqref{eq:je_dual_tensor} if and only if the following holds:
\begin{align*}
    \sum_{i=1}^{n^2} \mathsf{w}_i = d^2, \\
    \mathsf{a}_{i}^\top (\sum_{j=1}^{n^2}  \mathsf{w}_{j} \mathsf{a}_{j} \mathsf{a}_{j}^\top)^{-1}  \mathsf{a}_{i}= 1, \mathrm{if}~ \mathsf{w}_{i} \neq 0, \\
   \mathsf{a}_{i}^\top (\sum_{j=1}^{n^2}  \mathsf{w}_{j} \mathsf{a}_{j} \mathsf{a}_{j}^\top)^{-1}  \mathsf{a}_{i}< 1, \mathrm{if}~ \mathsf{w}_{i} = 0. 
\end{align*}
\end{lemma}

\begin{definition}[$(1+\epsilon)$-approximate solution, tensor version of Definition~\ref{def:approx_sol}]\label{def:approx_sol_tensor}
For any $\epsilon > 0$, we say that $\mathsf{w} \in \R^{n^2}_{\geq 0}$ is a $(1+\epsilon)$-approximate solution to program~\eqref{eq:dual_opt} if $w$ satisfies
\begin{align*}
    \sum_{i=1}^{n^2} \mathsf{w}_i = (1 \pm \epsilon)d^2, \\
    \mathsf{a}_{i}^\top (\sum_{j=1}^{n^2}  \mathsf{w}_{j} \mathsf{a}_{j} \mathsf{a}_{j}^\top)^{-1}  \mathsf{a}_{i}< 1 + \epsilon, \forall i \in [n^2], 
\end{align*}
\end{definition}

Next, the following lemma states that an approximate solution is a good approximation of the exact John Ellipsoid. For the tensor version, the exact John Ellipsoid $E^*$ satisfies $E^* \subseteq P \subseteq d \cdot E^*$.
\begin{lemma}[Approximate John Ellipsoid, tensor version of Lemma~\ref{lem:approx_john_ell}]\label{lem:approx_john_ell_tensor} 
Let $P := \{x \in \R^d : - \mathbf 1_n \leq Ax \leq \mathbf 1_n\}$ be a symmetric convex polytope. Let $\mathsf{w} \in \R^{n^2}_{\geq 0}$ be a $(1+\epsilon)$-approximate solution to program~\eqref{eq:je_dual_tensor}. Let
\begin{align*}
    E := \{ x \in \R^{d^2} : x^\top (\mathsf{A}^\top \diag(\mathsf{w}) \mathsf{A}) x \leq 1\}.
\end{align*}
Then 
\begin{align*}
    \frac{1}{\sqrt{1+\epsilon}} \cdot E \subseteq P \subseteq \sqrt{(1+\epsilon)} \cdot d \cdot E.
\end{align*}
Moreover,
\begin{align*}
    \vol(\frac{1}{\sqrt{1+\epsilon}}E) \geq e^{-\epsilon d^2} \cdot \vol(E^*).
\end{align*}
\end{lemma}

\subsection{Decomposition of Tensor Version}
\begin{algorithm*}[!t]
\caption{Quantum Algorithm for approximating John Ellipsoid (Tensor Version)}\label{alg:main_tensor} 
    \begin{algorithmic}[1]
    \Procedure{\textsc{ApproxTensorJE}}{$A_1, A_2 \in \R^{n \times d}, \epsilon \in (0,0.5)$} 
    \State \Comment{Symmetric convex polytope defined as $\{ x \in \R^{d^2}: -\mathbf{1}_{n^2} \leq (A_1 \otimes A_2) x \leq \mathbf{1}_{n^2}\}$}
    \State \Comment{Target approximation error $\epsilon \in (0,1)$}
    \State $w_1 \gets \textsc{ApproxJE}(A_1, \epsilon$) \Comment{Algorithm~\ref{alg:main}}
    \State $w_2 \gets \textsc{ApproxJE}(A_2, \epsilon$)
    \State $\mathsf{w} \gets w_1 \otimes w_2$ \Comment{Lemma~\ref{lem:approx-guar_tensor}}
    \State \Return $\mathsf{w}$ 
    \EndProcedure
    \end{algorithmic}
\end{algorithm*}

\begin{lemma} \label{lem:tensor_optimality}
    There exists unique $\mathsf{w} \in \R^{n^2}$ satisfying the following conditions
\begin{align*}
    \sum_{i=1}^{n^2} \mathsf{w}_i = d^2, \\
    \mathsf{a}_{i}^\top (\sum_{j=1}^{n^2}  \mathsf{w}_{j} \mathsf{a}_{j} \mathsf{a}_{j}^\top)^{-1}  \mathsf{a}_{i}= 1, \mathrm{if}~ \mathsf{w}_{i} \neq 0, \\
   \mathsf{a}_{i}^\top (\sum_{j=1}^{n^2}  \mathsf{w}_{j} \mathsf{a}_{j} \mathsf{a}_{j}^\top)^{-1}  \mathsf{a}_{i}< 1, \mathrm{if}~ \mathsf{w}_{i} = 0. 
\end{align*}
\end{lemma}
\begin{proof}
    By Lemma~65 of~\cite{ls20a}, the $\ell_\infty$ Lewis weight is unique and it uniquely defines the John ellipsoid of a given symmetric convex polytope. Then there must exist unique $\mathsf{w} \in \R^{d^2}$ satisfying the optimality condition in Lemma~\ref{lem:opt_cond_tensor}.
\end{proof}

\begin{lemma}\label{lem:w_1_optimality}
    There exists unique $w_1 \in \R^{n}$ satisfying the following conditions
\begin{align*}
    \sum_{i_1=1}^n w_{1,i_1} = d, \\
    a_{1, i_1}^\top (\sum_{j_1=1}^n w_{1,j_1} a_{1,j_1} a_{1,j_1}^\top)^{-1} a_{1, i_1} = 1, \mathrm{if}~ w_{1,{i_1}} \neq 0, \\
    a_{1, i_1}^\top (\sum_{j_1=1}^n w_{1,j_1} a_{1,j_1} a_{1,j_1}^\top)^{-1} a_{1, i_1} < 1, \mathrm{if}~ w_{1,{i_1}} = 0.
\end{align*}
\end{lemma}

\begin{lemma} \label{lem:w_2_optimality}
    There exists unique $w_2 \in \R^{n}$ satisfying the following conditions
\begin{align*}
    \sum_{i_2=1}^n w_{2,i_2} = d, \\
    a_{1, i_2}^\top (\sum_{j_2=1}^n w_{2,j_2} a_{2,j_2} a_{2,j_2}^\top)^{-1} a_{i_2} = 1, \mathrm{if}~ w_{2, i_2} \neq 0, \\
    a_{1, i_2}^\top (\sum_{j_2=1}^n w_{2,j_2} a_{2,j_2} a_{2,j_2}^\top)^{-1} a_{1, i_2} < 1, \mathrm{if}~ w_{2,i_2} = 0.
\end{align*}
\end{lemma}

\begin{lemma}\label{lem:w_decompose}
    Let $\mathsf{w}$ be defined in Lemma~\ref{lem:tensor_optimality}. Let $w_1$, $w_2$ be defined in Lemma~\ref{lem:w_1_optimality}, \ref{lem:w_2_optimality}, respectively. Then we have $\mathsf{w} = w_1 \otimes w_2$.
\end{lemma}

\begin{proof}

Firstly, we multiply the first conditions in Lemma~\ref{lem:w_1_optimality} and Lemma~\ref{lem:w_2_optimality}. Then we have
\begin{align*}
\sum_{i_1 = 1}^n \sum_{i_2 = 1}^n w_{1,i_1} w_{2, i_2} = d^2.
\end{align*}

Now by Fact~\ref{fact:rank1}, we have
\begin{align*}
&~ (a_{1,i_1}^\top (A_1^\top \diag (w_1) A_1)^{-1} a_{1,i_1}) \cdot (a_{2,i_2}^\top (A_2^\top \diag (w_2) A_2)^{-1} a_{2,i_2}) \\
= &~ (a_{1,i_1}^\top (\sum_{j_1=1}^n w_{1,j_1} a_{1,j_1} a_{1,j_1}^\top)^{-1} a_{1,i_1})
(a_{2,i_2}^\top (\sum_{j_2=1}^n w_{2,j_2} a_{2,j_2} a_{2,j_2}^\top)^{-1} a_{2, i_2}) 
\end{align*}
and
\begin{align*}
    &~ \mathsf{a}_{i_1+(i_2-1)n}^\top ( \mathsf{A} \diag(w_1 \otimes w_2) \mathsf{A} )^{-1} \mathsf{a}_{i_1+(i_2-1)n}
    \\ = &~
    \mathsf{a}_{i_1+(i_2-1)n}^\top (\sum_{j_1=1}^{n} \sum_{j_2=1}^n w_{j_1}w_{j_2} \mathsf{a}_{j_1+(j_2-1)n} \mathsf{a}_{j_1+(j_2-1)n}^\top)^{-1}  \mathsf{a}_{i_1+(i_2-1)n}.
\end{align*}

Next, for each pair of $i_1, i_2$, we multiply the corresponding second or third conditions. Then we have
\begin{align*}
&~ (a_{1,i_1}^\top (A_1^\top \diag (w_1) A_1)^{-1} a_{1,i_1}) \cdot (a_{2,i_2}^\top (A_2^\top \diag (w_2) A_2)^{-1} a_{2,i_2}) \\
= & ~ (a_{1,i_1}^\top (A_1^\top \diag (w_1) A_1)^{-1} a_{1,i_1}) \otimes (a_{2,i_2}^\top (A_2^\top \diag (w_2) A_2)^{-1} a_{2,i_2}) \\
= & ~ (a_{1,i_1}^\top \otimes a_{2, i_2}^\top) ((A_1^\top \diag (w_1) A_1 )^{-1} \otimes (A_2^\top \diag (w_2) A_2)^{-1}) (a_{1,i_1} \otimes a_{2, i_2}) \\
= & ~ (a_{1,i_1} \otimes a_{2, i_2})^\top ((A_1^\top \diag (w_1) A_1 ) \otimes (A_2^\top \diag (w_2) A_2))^{-1} (a_{1,i_1} \otimes a_{2, i_2}) \\
= & ~ (a_{1,i_1} \otimes a_{2, i_2})^\top ( (A_1 \otimes A_2)^\top (\diag(w_1) \otimes \diag(w_2)) (A_1 \otimes A_2) )^{-1} (a_{1,i_1} \otimes a_{2, i_2}) \\
= & ~ (a_{1,i_1} \otimes a_{2, i_2})^\top ( (A_1 \otimes A_2)^\top \diag(w_1 \otimes w_2) (A_1 \otimes A_2) )^{-1} (a_{1,i_1} \otimes a_{2, i_2}) \\
= & ~ \mathsf{a}_{i_1+(i_2-1)n}^\top ( \mathsf{A} \diag(w_1 \otimes w_2) \mathsf{A} )^{-1} \mathsf{a}_{i_1+(i_2-1)n}
\end{align*}
where the first step uses the fact that Kronecker product is equivalent to general product in scalar case, the second step uses Fact~\ref{fact:kronecker}, the third step follows from basic algebra, the fourth and fifth steps use Fact~\ref{fact:kronecker}, the last step is due to the definition of $\mathsf{a_{i,j}}$.

Hence we have
\begin{align*}
  &~  (a_{1, i_1}^\top (\sum_{j_1=1}^n w_{1,j_1} a_{1,j_1} a_{1,j_1}^\top)^{-1} a_{1, i_1})
(a_{2, i_2}^\top (\sum_{j_2=1}^n w_{2,j_2} a_{2,j_2} a_{2,j_2}^\top)^{-1} a_{2,i_2}) \\ 
= &~ \mathsf{a}_{i_1+(i_2-1)n}^\top (\sum_{j_1=1}^{n} \sum_{j_2=1}^n w_{1, j_1}w_{2, j_2} \mathsf{a}_{j_1+(j_2-1)n}\mathsf{a}_{j_1+(j_2-1)n}^\top)^{-1}  \mathsf{a}_{i_1+(i_2-1)n}.
\end{align*}
Note that the above equation is less than $1$ if either $w_{1,i_1} = 0$ or $w_{2,i_2} = 0$, and the above equation is equal to $1$ if both $w_{1,i_1}, w_{2,i_2}$ are not zero. Thus we have
\begin{align*}
    \mathsf{a}_{i_1+(i_2-1)n}^\top (\sum_{j_1=1}^{n} \sum_{j_2=1}^n w_{1, j_1}w_{2, j_2} \mathsf{a}_{j_1+(j_2-1)n}\mathsf{a}_{j_1+(j_2-1)n}^\top)^{-1}  \mathsf{a}_{i_1+(i_2-1)n} = 1,\mathrm{if}~ w_{1,i_1}w_{2,i_2} \neq 0, \\
    \mathsf{a}_{i_1+(i_2-1)n}^\top (\sum_{j_1=1}^{n} \sum_{j_2=1}^n w_{1, j_1}w_{2, j_2} \mathsf{a}_{j_1+(j_2-1)n}\mathsf{a}_{j_1+(j_2-1)n}^\top)^{-1}  \mathsf{a}_{i_1+(i_2-1)n} < 1,\mathrm{if}~ w_{1,i_1}w_{2,i_2} = 0.
\end{align*}
Thus we have the following conditions for $w_1, w_2$:
\begin{align}
\label{eq:tmp1}
    \sum_{i_1 = 1}^n \sum_{i_2 = 1}^n w_{1,i_1} w_{2, i_2} = d^2. \notag \\
    \mathsf{a}_{i_1+(i_2-1)n}^\top (\sum_{j_1=1}^{n} \sum_{j_2=1}^n w_{1, j_1}w_{2, j_2} \mathsf{a}_{j_1+(j_2-1)n}\mathsf{a}_{j_1+(j_2-1)n}^\top)^{-1}  \mathsf{a}_{i_1+(i_2-1)n} = 1,\mathrm{if}~ w_{1,i_1}w_{2,i_2} \neq 0, \\
    \mathsf{a}_{i_1+(i_2-1)n}^\top (\sum_{j_1=1}^{n} \sum_{j_2=1}^n w_{1, j_1}w_{2, j_2} \mathsf{a}_{j_1+(j_2-1)n}\mathsf{a}_{j_1+(j_2-1)n}^\top)^{-1}  \mathsf{a}_{i_1+(i_2-1)n} < 1,\mathrm{if}~ w_{1,i_1}w_{2,i_2} = 0 \notag.
\end{align}
By Lemma~\ref{lem:tensor_optimality}, there exists $\mathsf w \in \R^{n^2}$ which satisfies
\begin{align}
\label{eq:tmp2}
    \sum_{i=1}^{n^2} \mathsf{w}_i = d^2, \notag \\
    \mathsf{a}_{i}^\top (\sum_{j=1}^{n^2}  \mathsf{w}_{j} \mathsf{a}_{j} \mathsf{a}_{j}^\top)^{-1}  \mathsf{a}_{i}= 1, \mathrm{if}~ \mathsf{w}_{i} \neq 0, \\
   \mathsf{a}_{i}^\top (\sum_{j=1}^{n^2}  \mathsf{w}_{j} \mathsf{a}_{j} \mathsf{a}_{j}^\top)^{-1}  \mathsf{a}_{i}< 1, \mathrm{if}~ \mathsf{w}_{i} = 0 \notag. 
\end{align}
Comparing Eq.~\eqref{eq:tmp1} and Eq.~\eqref{eq:tmp2} and reindexing, we must have $\mathsf{w} = w_1 \otimes w_2$.
\end{proof}

The key message of Lemma~\ref{lem:w_decompose} is that solving the tensor version of the John ellipsoid problem can be decomposed to solving two ordinary John ellipsoid problems. Hence the techniques we developed can be used to speed up the tensor version of the John ellipsoid problem. 

\subsection{Algorithm and Approximation Guarantee}
We show that Algorithm~\ref{alg:main_tensor} can output a good approximation of the tensor version of the John ellipsoid problem.

\begin{lemma}[Approximation guarantee, tensor version of Lemma~\ref{lem:approx-guar}]\label{lem:approx-guar_tensor}
For any $\epsilon \in (0,0.5)$, with probability $1 - 1/\poly(n)$, Algorithm~\ref{alg:main_tensor} outputs a $(1+\epsilon)^2$-approximate solution.
\end{lemma}
\begin{proof}
By Lemma~\ref{lem:approx-guar}, $w_1$ and $w_2$ are $(1+\epsilon)$ approximations of the John ellipsoid problems defined by $A_1$ and $A_2$, respectively.
Hence we have
\begin{align*}
    \sum_{i_1=1}^n w_{1,i_1} = (1 \pm \epsilon)d, \\
    a_{1, i_1}^\top (\sum_{j_1=1}^n w_{1,j_1} a_{1,j_1} a_{1,j_1}^\top)^{-1} a_{1, i_1} \leq 1 + \epsilon, \forall i_1 \in [n], 
\end{align*}
and
\begin{align*}
    \sum_{i_2=1}^n w_{2,i_2} = (1 \pm \epsilon)d, \\
    a_{2, i_2}^\top (\sum_{j_2=1}^n w_{2,j_2} a_{2,j_2} a_{2,j_2}^\top)^{-1} a_{2, i_2} \leq 1 + \epsilon, \forall i_2 \in [n].
\end{align*}
Let $\mathsf{w} = w_1 \otimes w_2$. By the similar argument in the proof of Lemma~\ref{lem:w_decompose}, multiplying each pair from the above equations give
\begin{align*}
    \sum_{i=1}^{n^2}\mathsf{w}_{i}= (1 \pm \epsilon)^2d^2, \\
    \mathsf{a}_{i}^\top (\sum_{j=1}^{n^2} w_{j}\mathsf{a}_{j} \mathsf{a}_{j}^\top)^{-1}  \mathsf{a}_{i} \leq (1 + \epsilon)^2, \forall {i} \in [n^2].
\end{align*}
By Definition~\ref{def:approx_sol_tensor}, $\mathsf{w}$ is a $(1+\epsilon)^2$ approximate solution.
\end{proof}

\begin{theorem}[Main result, tensor version of Theorem~\ref{thm:main}]\label{thm:main_tensor}
Given $A_1, A_2 \in \R^{n \times d}$, let $\mathsf{A} = A_1 \otimes A_2$ and let $P := \{x \in \R^{d^2} : -\mathbf{1}_{n^2} \leq \mathsf{A} x \leq \mathbf{1}_{n^2}\}$ be a symmetric convex polytope.
    For all $\epsilon \in (0,0.5)$, Algorithm~\ref{alg:main_tensor} outputs $\mathsf{w} \in \R^{n^2}$ such that the ellipsoid $E := \{x \in \R^{d^2}: x^\top(\mathsf{A}^\top\diag(\mathsf{w})\mathsf{A})x \leq 1\}$ satisfies 
    \begin{itemize}
        \item $ \frac{1}{1+\epsilon} \cdot E \subseteq P \subseteq (1+\epsilon)d \cdot E. $
        \item $\vol(\frac{1}{1+\epsilon}E) \geq e^{-\epsilon^2 d^2} \cdot \vol(E^*)$ where $E^*$ is the exact John ellipsoid of $P$.
        \item It runs in $\wt{O}((\epsilon^{-2}\sqrt{n}d^{1.5}+\epsilon^{-3}d^\omega)$ time.
        \item It makes $\Tilde{O}(\epsilon^{-2}\sqrt{nd})$ queries to $A_1$ and $A_2$.
        \item The success probability is $1 - 1/\poly(n)$.
    \end{itemize}
\end{theorem}

\begin{proof}
    The correctness of main result directly follows from Lemma~\ref{lem:approx-guar_tensor} and Lemma~\ref{lem:approx_john_ell_tensor}.
    The time complexity and number of queries follow from Lemma~\ref{lem:qua_fix_point_iter}.
\end{proof}  
\section{Conclusion}\label{sec:conc}
In conclusion, we have developed the first quantum algorithm for approximating the John Ellipsoid of a symmetric convex polytope that achieves quadratic speedup and leads to a sublinear running time. We use quantum sketching techniques to speedup the classical state-of-the-art $O(\nnz(A)+d^\omega)$ time~\cite{song2022faster} and achieve $O(\sqrt{n}d^{1.5} + d^\omega)$ time. Quantum computing has emerged as a powerful tool for designing exceptionally fast optimization algorithms in high-dimensional settings. Future work could explore the extension of these techniques to other geometric optimization problems and further enhancing the efficiency and practicality of quantum algorithms.

\newpage

\ifdefined\isarxiv
\bibliographystyle{alpha}
\bibliography{ref}
\else
\bibliography{ref}

\newcommand{\etalchar}[1]{$^{#1}$}
\begin{thebibliography}{DMIMW12}

\bibitem[ADW22]{apers2022quantum}
Simon Apers and Ronald De~Wolf.
\newblock Quantum speedup for graph sparsification, cut approximation, and
  laplacian solving.
\newblock {\em SIAM Journal on Computing}, 51(6):1703--1742, 2022.

\bibitem[ADW{\etalchar{+}}24]{alman2024asymmetryyieldsfastermatrix}
Josh Alman, Ran Duan, Virginia~Vassilevska Williams, Yinzhan Xu, Zixuan Xu, and
  Renfei Zhou.
\newblock More asymmetry yields faster matrix multiplication, 2024.

\bibitem[AG23]{apers2023quantum}
Simon Apers and Sander Gribling.
\newblock Quantum speedups for linear programming via interior point methods.
\newblock {\em arXiv preprint arXiv:2311.03215}, 2023.

\bibitem[Atw69]{atwood1969optimal}
Corwin~L Atwood.
\newblock Optimal and efficient designs of experiments.
\newblock {\em The Annals of Mathematical Statistics}, pages 1570--1602, 1969.

\bibitem[AW21]{alman2021refined}
Josh Alman and Virginia~Vassilevska Williams.
\newblock A refined laser method and faster matrix multiplication.
\newblock In {\em Proceedings of the 2021 ACM-SIAM Symposium on Discrete
  Algorithms (SODA)}, pages 522--539. SIAM, 2021.

\bibitem[AZLSW17]{allen2017near}
Zeyuan Allen-Zhu, Yuanzhi Li, Aarti Singh, and Yining Wang.
\newblock Near-optimal design of experiments via regret minimization.
\newblock In {\em International Conference on Machine Learning}, pages
  126--135. PMLR, 2017.

\bibitem[Bal91]{ball1991volume}
Keith Ball.
\newblock Volume ratios and a reverse isoperimetric inequality.
\newblock {\em Journal of the London Mathematical Society}, 2(2):351--359,
  1991.

\bibitem[Bal01]{ball2001convex}
Keith Ball.
\newblock Convex geometry and functional analysis.
\newblock In {\em Handbook of the geometry of Banach spaces}, volume~1, pages
  161--194. Elsevier, 2001.

\bibitem[Bar98]{barthe1998extremal}
Franck Barthe.
\newblock An extremal property of the mean width of the simplex.
\newblock {\em Mathematische Annalen}, 310:685--693, 1998.

\bibitem[BCBK12]{bubeck2012towards}
S{\'e}bastien Bubeck, Nicolo Cesa-Bianchi, and Sham~M Kakade.
\newblock Towards minimax policies for online linear optimization with bandit
  feedback.
\newblock In {\em Conference on Learning Theory}, pages 41--1. JMLR Workshop
  and Conference Proceedings, 2012.

\bibitem[BDCG{\etalchar{+}}20]{ben2020symmetries}
Shalev Ben-David, Andrew~M Childs, Andr{\'a}s Gily{\'e}n, William Kretschmer,
  Supartha Podder, and Daochen Wang.
\newblock Symmetries, graph properties, and quantum speedups.
\newblock In {\em 2020 IEEE 61st Annual Symposium on Foundations of Computer
  Science (FOCS)}, pages 649--660. IEEE, 2020.

\bibitem[BKL{\etalchar{+}}19]{brandao2019quantum}
Fernando~GSL Brand{\~a}o, Amir Kalev, Tongyang Li, Cedric Yen-Yu Lin, Krysta~M
  Svore, and Xiaodi Wu.
\newblock Quantum sdp solvers: Large speed-ups, optimality, and applications to
  quantum learning.
\newblock In {\em 46th International Colloquium on Automata, Languages, and
  Programming (ICALP 2019)}. Schloss-Dagstuhl-Leibniz Zentrum f{\"u}r
  Informatik, 2019.

\bibitem[BLM89]{bourgain1989approximation}
Jean Bourgain, Joram Lindenstrauss, and Vitali Milman.
\newblock Approximation of zonoids by zonotopes.
\newblock 1989.

\bibitem[BLSS20]{blss20}
Jan van~den Brand, Yin~Tat Lee, Aaron Sidford, and Zhao Song.
\newblock Solving tall dense linear programs in nearly linear time.
\newblock In {\em Proceedings of the 52nd Annual ACM SIGACT Symposium on Theory
  of Computing}, pages 775--788, 2020.

\bibitem[Bra20]{b20}
Jan van~den Brand.
\newblock A deterministic linear program solver in current matrix
  multiplication time.
\newblock In {\em Proceedings of the Fourteenth Annual ACM-SIAM Symposium on
  Discrete Algorithms}, pages 259--278. SIAM, 2020.

\bibitem[BS17]{brandao2017quantum}
Fernando~GSL Brandao and Krysta~M Svore.
\newblock Quantum speed-ups for solving semidefinite programs.
\newblock In {\em 2017 IEEE 58th Annual Symposium on Foundations of Computer
  Science (FOCS)}, pages 415--426. IEEE, 2017.

\bibitem[BW14]{boutsidis2014optimal}
Christos Boutsidis and David~P Woodruff.
\newblock Optimal cur matrix decompositions.
\newblock In {\em Proceedings of the forty-sixth annual ACM symposium on Theory
  of computing}, pages 353--362, 2014.

\bibitem[CCLY19]{cohen2019near}
Michael~B Cohen, Ben Cousins, Yin~Tat Lee, and Xin Yang.
\newblock A near-optimal algorithm for approximating the john ellipsoid.
\newblock In {\em Conference on Learning Theory}, pages 849--873. PMLR, 2019.

\bibitem[CdW21]{cd21}
Yanlin Chen and Ronald de~Wolf.
\newblock Quantum algorithms and lower bounds for linear regression with norm
  constraints.
\newblock {\em arXiv preprint arXiv:2110.13086}, 2021.

\bibitem[CDWY18]{chen2018fast}
Yuansi Chen, Raaz Dwivedi, Martin~J Wainwright, and Bin Yu.
\newblock Fast mcmc sampling algorithms on polytopes.
\newblock {\em Journal of Machine Learning Research}, 19(55):1--86, 2018.

\bibitem[CGJ19]{cgj18}
Shantanav Chakraborty, Andr\'{a}s Gily\'{e}n, and Stacey Jeffery.
\newblock {The Power of Block-Encoded Matrix Powers: Improved Regression
  Techniques via Faster Hamiltonian Simulation}.
\newblock In Christel Baier, Ioannis Chatzigiannakis, Paola Flocchini, and
  Stefano Leonardi, editors, {\em 46th International Colloquium on Automata,
  Languages, and Programming (ICALP 2019)}, volume 132 of {\em Leibniz
  International Proceedings in Informatics (LIPIcs)}, pages 33:1--33:14,
  Dagstuhl, Germany, 2019. Schloss Dagstuhl -- Leibniz-Zentrum f{\"u}r
  Informatik.

\bibitem[CLM{\etalchar{+}}15]{cohen2015uniform}
Michael~B Cohen, Yin~Tat Lee, Cameron Musco, Christopher Musco, Richard Peng,
  and Aaron Sidford.
\newblock Uniform sampling for matrix approximation.
\newblock In {\em Proceedings of the 2015 conference on innovations in
  theoretical computer science}, pages 181--190, 2015.

\bibitem[CLS21]{cls19}
Michael~B Cohen, Yin~Tat Lee, and Zhao Song.
\newblock Solving linear programs in the current matrix multiplication time.
\newblock {\em Journal of the ACM (JACM)}, 68(1):1--39, 2021.

\bibitem[CMD20]{casares2020quantum}
Pablo~AM Casares and Miguel~Angel Martin-Delgado.
\newblock A quantum interior-point predictor--corrector algorithm for linear
  programming.
\newblock {\em Journal of physics A: Mathematical and Theoretical},
  53(44):445305, 2020.

\bibitem[CP15]{cohen2015lp}
Michael~B Cohen and Richard Peng.
\newblock Lp row sampling by lewis weights.
\newblock In {\em Proceedings of the forty-seventh annual ACM symposium on
  Theory of computing}, pages 183--192, 2015.

\bibitem[CW17]{clarkson2017low}
Kenneth~L Clarkson and David~P Woodruff.
\newblock Low-rank approximation and regression in input sparsity time.
\newblock {\em Journal of the ACM (JACM)}, 63(6):1--45, 2017.

\bibitem[Dan51]{dan51}
George~B Dantzig.
\newblock Maximization of a linear function of variables subject to linear
  inequalities.
\newblock {\em Activity analysis of production and allocation}, 13:339--347,
  1951.

\bibitem[DAST08]{damla2008linear}
S~Damla~Ahipasaoglu, Peng Sun, and Michael~J Todd.
\newblock Linear convergence of a modified frank--wolfe algorithm for computing
  minimum-volume enclosing ellipsoids.
\newblock {\em Optimisation Methods and Software}, 23(1):5--19, 2008.

\bibitem[DMIMW12]{drineas2012fast}
Petros Drineas, Malik Magdon-Ismail, Michael~W Mahoney, and David~P Woodruff.
\newblock Fast approximation of matrix coherence and statistical leverage.
\newblock {\em The Journal of Machine Learning Research}, 13(1):3475--3506,
  2012.

\bibitem[DS08]{ds08}
Samuel~I Daitch and Daniel~A Spielman.
\newblock Faster approximate lossy generalized flow via interior point
  algorithms.
\newblock In {\em Proceedings of the fortieth annual ACM symposium on Theory of
  computing}, pages 451--460, 2008.

\bibitem[DSSW18]{dssw18}
Huaian Diao, Zhao Song, Wen Sun, and David Woodruff.
\newblock Sketching for kronecker product regression and p-splines.
\newblock In {\em International Conference on Artificial Intelligence and
  Statistics}, pages 1299--1308. PMLR, 2018.

\bibitem[DSW22]{deng2022discrepancy}
Yichuan Deng, Zhao Song, and Omri Weinstein.
\newblock Discrepancy minimization in input-sparsity time.
\newblock {\em arXiv preprint arXiv:2210.12468}, 2022.

\bibitem[Gro96]{grover1996fast}
Lov~K Grover.
\newblock A fast quantum mechanical algorithm for database search.
\newblock In {\em Proceedings of the twenty-eighth annual ACM symposium on
  Theory of computing}, pages 212--219, 1996.

\bibitem[GS22]{gs22}
Yuzhou Gu and Zhao Song.
\newblock A faster small treewidth sdp solver.
\newblock {\em arXiv preprint arXiv:2211.06033}, 2022.

\bibitem[GST22]{gilyen2022improved}
Andr{\'a}s Gily{\'e}n, Zhao Song, and Ewin Tang.
\newblock An improved quantum-inspired algorithm for linear regression.
\newblock {\em Quantum}, 6:754, 2022.

\bibitem[GSWY23]{gswy23}
Yeqi Gao, Zhao Song, Weixin Wang, and Junze Yin.
\newblock A fast optimization view: Reformulating single layer attention in llm
  based on tensor and svm trick, and solving it in matrix multiplication time.
\newblock {\em arXiv preprint arXiv:2309.07418}, 2023.

\bibitem[GSY23]{gsy23_hyper}
Yeqi Gao, Zhao Song, and Junze Yin.
\newblock An iterative algorithm for rescaled hyperbolic functions regression.
\newblock {\em arXiv preprint arXiv:2305.00660}, 2023.

\bibitem[GSYZ23]{gao2023fast}
Yeqi Gao, Zhao Song, Xin Yang, and Ruizhe Zhang.
\newblock Fast quantum algorithm for attention computation.
\newblock {\em arXiv preprint arXiv:2307.08045}, 2023.

\bibitem[GSZ23]{gsz23}
Yuzhou Gu, Zhao Song, and Lichen Zhang.
\newblock A nearly-linear time algorithm for structured support vector
  machines.
\newblock {\em arXiv preprint arXiv:2307.07735}, 2023.

\bibitem[GSZ24]{gao2024quantumspeedupspectralapproximation}
Yeqi Gao, Zhao Song, and Ruizhe Zhang.
\newblock Quantum speedup for spectral approximation of kronecker products,
  2024.

\bibitem[Hal20]{halder2020smallest}
Abhishek Halder.
\newblock Smallest ellipsoid containing $ p $-sum of ellipsoids with
  application to reachability analysis.
\newblock {\em IEEE Transactions on Automatic Control}, 66(6):2512--2525, 2020.

\bibitem[HJS{\etalchar{+}}22a]{huang2022faster}
Baihe Huang, Shunhua Jiang, Zhao Song, Runzhou Tao, and Ruizhe Zhang.
\newblock A faster quantum algorithm for semidefinite programming via robust
  ipm framework.
\newblock {\em arXiv preprint arXiv:2207.11154}, 2022.

\bibitem[HJS{\etalchar{+}}22b]{huang2022solving}
Baihe Huang, Shunhua Jiang, Zhao Song, Runzhou Tao, and Ruizhe Zhang.
\newblock Solving sdp faster: A robust ipm framework and efficient
  implementation.
\newblock In {\em 2022 IEEE 63rd Annual Symposium on Foundations of Computer
  Science (FOCS)}, pages 233--244. IEEE, 2022.

\bibitem[HK16]{hazan2016volumetric}
Elad Hazan and Zohar Karnin.
\newblock Volumetric spanners: an efficient exploration basis for learning.
\newblock {\em Journal of Machine Learning Research}, 2016.

\bibitem[JKL{\etalchar{+}}20]{jkl+20}
Haotian Jiang, Tarun Kathuria, Yin~Tat Lee, Swati Padmanabhan, and Zhao Song.
\newblock A faster interior point method for semidefinite programming.
\newblock In {\em 2020 IEEE 61st annual symposium on foundations of computer
  science (FOCS)}, pages 910--918. IEEE, 2020.

\bibitem[Joh48]{john1948extremum}
Fritz John.
\newblock Extremum problems with inequalities as subsidiary conditions, studies
  and essays presented to r. courant on his 60th birthday, january 8, 1948,
  1948.

\bibitem[JSWZ21]{jswz21}
Shunhua Jiang, Zhao Song, Omri Weinstein, and Hengjie Zhang.
\newblock A faster algorithm for solving general lps.
\newblock In {\em Proceedings of the 53rd Annual ACM SIGACT Symposium on Theory
  of Computing}, pages 823--832, 2021.

\bibitem[Kar84]{k84}
Narendra Karmarkar.
\newblock A new polynomial-time algorithm for linear programming.
\newblock In {\em Proceedings of the sixteenth annual ACM symposium on Theory
  of computing}, pages 302--311, 1984.

\bibitem[Kha79]{k79}
Leonid~Genrikhovich Khachiyan.
\newblock A polynomial algorithm in linear programming.
\newblock In {\em Doklady Akademii Nauk}, volume 244, pages 1093--1096. Russian
  Academy of Sciences, 1979.

\bibitem[Kha96]{khachiyan1996rounding}
Leonid~G Khachiyan.
\newblock Rounding of polytopes in the real number model of computation.
\newblock {\em Mathematics of Operations Research}, 21(2):307--320, 1996.

\bibitem[KY05]{kumar2005minimum}
Piyush Kumar and E~Alper Yildirim.
\newblock Minimum-volume enclosing ellipsoids and core sets.
\newblock {\em Journal of Optimization Theory and applications}, 126(1):1--21,
  2005.

\bibitem[LBD20]{lin2020nonnegative}
Chia-Hsiang Lin and Jos{\'e}~M Bioucas-Dias.
\newblock Nonnegative blind source separation for ill-conditioned mixtures via
  john ellipsoid.
\newblock {\em IEEE Transactions on Neural Networks and Learning Systems},
  32(5):2209--2223, 2020.

\bibitem[LFN18]{lu2018relatively}
Haihao Lu, Robert~M Freund, and Yurii Nesterov.
\newblock Relatively smooth convex optimization by first-order methods, and
  applications.
\newblock {\em SIAM Journal on Optimization}, 28(1):333--354, 2018.

\bibitem[LG14]{le2014powers}
Fran{\c{c}}ois Le~Gall.
\newblock Powers of tensors and fast matrix multiplication.
\newblock In {\em Proceedings of the 39th international symposium on symbolic
  and algebraic computation}, pages 296--303, 2014.

\bibitem[LLS{\etalchar{+}}24]{lls+24}
Xiaoyu Li, Yingyu Liang, Zhenmei Shi, Zhao Song, and Junwei Yu.
\newblock Fast john ellipsoid computation with differential privacy
  optimization, 2024.

\bibitem[LS14]{ls14}
Yin~Tat Lee and Aaron Sidford.
\newblock Path finding methods for linear programming: Solving linear programs
  in o (vrank) iterations and faster algorithms for maximum flow.
\newblock In {\em 2014 IEEE 55th Annual Symposium on Foundations of Computer
  Science}, pages 424--433. IEEE, 2014.

\bibitem[LS20a]{ls20a}
Yin~Tat Lee and Aaron Sidford.
\newblock Solving linear programs with sqrt(rank) linear system solves, 2020.

\bibitem[LS20b]{liu2020faster}
Yang~P Liu and Aaron Sidford.
\newblock Faster energy maximization for faster maximum flow.
\newblock In {\em Proceedings of the 52nd Annual ACM SIGACT Symposium on Theory
  of Computing}, pages 803--814, 2020.

\bibitem[LSZ19]{lsz19}
Yin~Tat Lee, Zhao Song, and Qiuyi Zhang.
\newblock Solving empirical risk minimization in the current matrix
  multiplication time.
\newblock In {\em Conference on Learning Theory}, pages 2140--2157. PMLR, 2019.

\bibitem[LSZ20]{liu2020breaking}
S~Cliff Liu, Zhao Song, and Hengjie Zhang.
\newblock Breaking the n-pass barrier: A streaming algorithm for maximum weight
  bipartite matching.
\newblock {\em arXiv preprint arXiv:2009.06106}, 2020.

\bibitem[Lut93]{lutwak1993brunn}
Erwin Lutwak.
\newblock The brunn-minkowski-firey theory. i. mixed volumes and the minkowski
  problem.
\newblock {\em Journal of Differential Geometry}, 38(1):131--150, 1993.

\bibitem[LYZ05]{lutwak2005john}
Erwin Lutwak, Deane Yang, and Gaoyong Zhang.
\newblock John ellipsoids.
\newblock {\em Proceedings of the London Mathematical Society}, 90(2):497--520,
  2005.

\bibitem[LZ17]{liu2017fast}
Yang Liu and Shengyu Zhang.
\newblock Fast quantum algorithms for least squares regression and statistic
  leverage scores.
\newblock {\em Theoretical Computer Science}, 657:38--47, 2017.

\bibitem[Mad13]{madry2013navigating}
Aleksander Madry.
\newblock Navigating central path with electrical flows: From flows to
  matchings, and back.
\newblock In {\em 2013 IEEE 54th Annual Symposium on Foundations of Computer
  Science}, pages 253--262. IEEE, 2013.

\bibitem[Mad16]{madry2016computing}
Aleksander Madry.
\newblock Computing maximum flow with augmenting electrical flows.
\newblock In {\em 2016 IEEE 57th Annual Symposium on Foundations of Computer
  Science (FOCS)}, pages 593--602. IEEE, 2016.

\bibitem[MAST23]{mohammadisiahroudi2023quantum}
Mohammadhossein Mohammadisiahroudi, Brandon Augustino, Pouya Sampourmahani, and
  Tam{\'a}s Terlaky.
\newblock Quantum computing inspired iterative refinement for semidefinite
  optimization.
\newblock {\em arXiv preprint arXiv:2312.11253}, 2023.

\bibitem[Mon15]{montanaro2015quantum}
Ashley Montanaro.
\newblock Quantum speedup of monte carlo methods.
\newblock {\em Proceedings of the Royal Society A: Mathematical, Physical and
  Engineering Sciences}, 471(2181):20150301, 2015.

\bibitem[NN94]{nesterov1994interior}
Yurii Nesterov and Arkadii Nemirovskii.
\newblock {\em Interior-point polynomial algorithms in convex programming}.
\newblock SIAM, 1994.

\bibitem[NN13]{nelson2013osnap}
Jelani Nelson and Huy~L Nguy{\^e}n.
\newblock Osnap: Faster numerical linear algebra algorithms via sparser
  subspace embeddings.
\newblock In {\em 2013 ieee 54th annual symposium on foundations of computer
  science}, pages 117--126. IEEE, 2013.

\bibitem[NTZ13]{nikolov2013geometry}
Aleksandar Nikolov, Kunal Talwar, and Li~Zhang.
\newblock The geometry of differential privacy: the sparse and approximate
  cases.
\newblock In {\em Proceedings of the forty-fifth annual ACM symposium on Theory
  of computing}, pages 351--360, 2013.

\bibitem[Pra14]{prakash2014quantum}
Anupam Prakash.
\newblock {\em Quantum algorithms for linear algebra and machine learning}.
\newblock University of California, Berkeley, 2014.

\bibitem[QSZZ23]{qszz23}
Lianke Qin, Zhao Song, Lichen Zhang, and Danyang Zhuo.
\newblock An online and unified algorithm for projection matrix vector
  multiplication with application to empirical risk minimization.
\newblock In {\em AISTATS}, 2023.

\bibitem[Ren88]{r88}
James Renegar.
\newblock A polynomial-time algorithm, based on newton's method, for linear
  programming.
\newblock {\em Mathematical programming}, 40(1):59--93, 1988.

\bibitem[Sch18]{schild2018almost}
Aaron Schild.
\newblock An almost-linear time algorithm for uniform random spanning tree
  generation.
\newblock In {\em Proceedings of the 50th Annual ACM SIGACT Symposium on Theory
  of Computing}, pages 214--227, 2018.

\bibitem[SF04]{sun2004computation}
Peng Sun and Robert~M Freund.
\newblock Computation of minimum-volume covering ellipsoids.
\newblock {\em Operations Research}, 52(5):690--706, 2004.

\bibitem[Sha23]{sha23}
Changpeng Shao.
\newblock Quantum speedup of leverage score sampling and its application, 2023.

\bibitem[SS08]{spielman2008graph}
Daniel~A Spielman and Nikhil Srivastava.
\newblock Graph sparsification by effective resistances.
\newblock In {\em Proceedings of the fortieth annual ACM symposium on Theory of
  computing}, pages 563--568, 2008.

\bibitem[SWYZ21]{swyz21}
Zhao Song, David Woodruff, Zheng Yu, and Lichen Zhang.
\newblock Fast sketching of polynomial kernels of polynomial degree.
\newblock In {\em International Conference on Machine Learning}, pages
  9812--9823. PMLR, 2021.

\bibitem[SWYZ23]{swyz23}
Zhao Song, Yitan Wang, Zheng Yu, and Lichen Zhang.
\newblock Sketching for first order method: Efficient algorithm for
  low-bandwidth channel and vulnerability.
\newblock In {\em ICML}, 2023.

\bibitem[SWZ19]{song2019relative}
Zhao Song, David~P Woodruff, and Peilin Zhong.
\newblock Relative error tensor low rank approximation.
\newblock In {\em Proceedings of the Thirtieth Annual ACM-SIAM Symposium on
  Discrete Algorithms}, pages 2772--2789. SIAM, 2019.

\bibitem[SXZ22]{sxz23}
Zhao Song, Zhaozhuo Xu, and Lichen Zhang.
\newblock Speeding up sparsification using inner product search data
  structures.
\newblock {\em arXiv preprint arXiv:2204.03209}, 2022.

\bibitem[SY21]{sy21}
Zhao Song and Zheng Yu.
\newblock Oblivious sketching-based central path method for linear programming.
\newblock In {\em International Conference on Machine Learning}, pages
  9835--9847. PMLR, 2021.

\bibitem[SYYZ22]{song2022faster}
Zhao Song, Xin Yang, Yuanyuan Yang, and Tianyi Zhou.
\newblock Faster algorithm for structured john ellipsoid computation.
\newblock {\em arXiv preprint arXiv:2211.14407}, 2022.

\bibitem[SYZ23a]{syz23}
Zhao Song, Junze Yin, and Lichen Zhang.
\newblock Solving attention kernel regression problem via pre-conditioner.
\newblock {\em arXiv preprint arXiv:2308.14304}, 2023.

\bibitem[SYZ23b]{song2023revisitingquantumalgorithmslinear}
Zhao Song, Junze Yin, and Ruizhe Zhang.
\newblock Revisiting quantum algorithms for linear regressions: Quadratic
  speedups without data-dependent parameters, 2023.

\bibitem[TLY24]{tang2024uncertainty}
Yukai Tang, Jean-Bernard Lasserre, and Heng Yang.
\newblock Uncertainty quantification of set-membership estimation in control
  and perception: Revisiting the minimum enclosing ellipsoid.
\newblock In {\em 6th Annual Learning for Dynamics \& Control Conference},
  pages 286--298. PMLR, 2024.

\bibitem[Tod16]{todd2016minimum}
Michael~J Todd.
\newblock {\em Minimum-volume ellipsoids: Theory and algorithms}.
\newblock SIAM, 2016.

\bibitem[Vai87]{v87}
Pravin~M Vaidya.
\newblock An algorithm for linear programming which requires o (((m+ n) n 2+(m+
  n) 1.5 n) l) arithmetic operations.
\newblock In {\em Proceedings of the nineteenth annual ACM symposium on Theory
  of computing}, pages 29--38, 1987.

\bibitem[Vai89]{v89}
Pravin~M Vaidya.
\newblock Speeding-up linear programming using fast matrix multiplication.
\newblock In {\em 30th annual symposium on foundations of computer science},
  pages 332--337. IEEE Computer Society, 1989.

\bibitem[Vem05]{vempala2005geometric}
Santosh Vempala.
\newblock Geometric random walks: a survey.
\newblock {\em Combinatorial and computational geometry}, 52(573-612):2, 2005.

\bibitem[WA08]{wocjan2008speedup}
Pawel Wocjan and Anura Abeyesinghe.
\newblock Speedup via quantum sampling.
\newblock {\em Physical Review A—Atomic, Molecular, and Optical Physics},
  78(4):042336, 2008.

\bibitem[Wan17]{wang2017quantum}
Guoming Wang.
\newblock Quantum algorithm for linear regression.
\newblock {\em Physical review A}, 96(1):012335, 2017.

\bibitem[Wil12]{williams2012multiplying}
Virginia~Vassilevska Williams.
\newblock Multiplying matrices faster than coppersmith-winograd.
\newblock In {\em Proceedings of the forty-fourth annual ACM symposium on
  Theory of computing}, pages 887--898, 2012.

\bibitem[Woo14]{w14}
David~P Woodruff.
\newblock Sketching as a tool for numerical linear algebra.
\newblock {\em Foundations and Trends{\textregistered} in Theoretical Computer
  Science}, 10(1--2):1--157, 2014.

\bibitem[WSP{\etalchar{+}}21]{wild2021quantum}
Dominik~S Wild, Dries Sels, Hannes Pichler, Cristian Zanoci, and Mikhail~D
  Lukin.
\newblock Quantum sampling algorithms, phase transitions, and computational
  complexity.
\newblock {\em Physical Review A}, 104(3):032602, 2021.

\bibitem[WXXZ24]{williams2024new}
Virginia~Vassilevska Williams, Yinzhan Xu, Zixuan Xu, and Renfei Zhou.
\newblock New bounds for matrix multiplication: from alpha to omega.
\newblock In {\em Proceedings of the 2024 Annual ACM-SIAM Symposium on Discrete
  Algorithms (SODA)}, pages 3792--3835. SIAM, 2024.

\bibitem[WYS17]{wang2017computationally}
Yining Wang, Adams~Wei Yu, and Aarti Singh.
\newblock On computationally tractable selection of experiments in
  measurement-constrained regression models.
\newblock {\em Journal of Machine Learning Research}, 18(143):1--41, 2017.

\bibitem[ZX14]{zou2014orlicz}
Du~Zou and Ge~Xiong.
\newblock Orlicz--john ellipsoids.
\newblock {\em Advances in Mathematics}, 265:132--168, 2014.

\end{thebibliography}
\bibliographystyle{alpha}

\fi

\newpage
\onecolumn




\end{document}